\documentclass[10pt,journal,twocolumn]{IEEEtran}
\usepackage{amsfonts,amssymb,amsmath,amsthm,bm}
\usepackage{algorithm}

\usepackage{algorithmic}
\usepackage[dvips]{graphicx}
\usepackage{float}
\usepackage{times}
\usepackage{color}
\usepackage{cite}
\usepackage{subfigure}
\newtheorem{remark}{Remark}
\newtheorem{theorem}{Theorem}
\newtheorem{example}{Example}
\DeclareMathOperator{\sgn}{sgn}

\newcommand{\bR}{\ensuremath{\mathbb R}}
\newcommand{\dB}{\text{dB}}

\newcommand{\bchi}{\bm{\chi}}
\newcommand{\bChi}{\bm{\chi}}
\newcommand{\bx}{\mathbf{x}}
\newcommand{\bz}{\mathbf{z}}

\newcommand{\bC}{\mathbf{C}}
\newcommand{\bX}{\mathbf{X}}
\newcommand{\bI}{\mathbf{I}}
\newcommand{\bS}{\mathbf{S}}
\newcommand{\bDelta}{\mathbf{\Delta}}

\begin{document}

%\title{Complementary Waveforms and Doppler-Induced Sidelobes}
\title{Coordinating Complementary Waveforms for Suppressing Range Sidelobes in a Doppler Band}
\author{Wenbing Dang, \textit{Member}, \textit{IEEE}, Ali Pezeshki, \textit{Member},
\textit{IEEE}, \\ Stephen D. Howard, William Moran,
\textit{Member}, \textit{IEEE}, \\ and Robert Calderbank, \textit{Fellow},
\textit{IEEE}}\maketitle

\let\thefootnote\relax\footnotetext{Submitted to the IEEE Transactions on Signal Processing, August 12, 2019. This work was supported in part by NSF under Grants
CCF-0916314, CCF-1018472, and CCF-1422658 and by AFOSR, AFRL,  and AOARD,  under contracts
FA9550-09-1-0518, FA2386-15-1-4066, FA2386-13-1-4080. Preliminary versions of parts of this paper were reported in the Forty Fifth Asilomar Conference on Signals, Systems, and Computers, Pacific Grove, CA, Nov. 6-9, 2011 (see \cite{Dang_Asilomar11}) and the 2012 IEEE Statistical Signal Processing Workshop, Ann Arbor, MI, Aug. 5-8, 2012 (see \cite{Dang_SSP12}).} 
\footnotetext{W. Dang was with the Department of Electrical and Computer
Engineering, Colorado State University,
Fort Collins, CO, 80523-1373, USA. He is now with Argo AI, Mountain View, CA, 94041, USA  (e-mail: dwb87514@gmail.com).} 
\footnotetext{A. Pezeshki (corresponding author) is with the Department of Electrical and Computer 
Engineering, and the Department of Mathematics, Colorado State University,
Fort Collins, CO, 80523-1373, USA (e-mail: Ali.Pezeshki@colostate.edu).}
\footnotetext{S. D. Howard is with the Defence Science and Technology Group, P.O. Box
1500, Edinburgh, SA, 5111, Australia (e-mail: Stephen.Howard@dsto.defence.gov.au).}
\footnotetext{W. Moran is with the Department of Electrical and Electronic Engineering, The University of Melbourne, Melbourne, VIC., 3010, Australia
(e-mail: wmoran@unimelb.edu.au).}
\footnotetext{R. Calderbank is with the Department of Electrical Engineering, the Department of Computer Science, and the Department of Mathematics at Duke University,
Durham, NC 27708 USA (e-mail: robert.calderbank@duke.edu).}
\begin{abstract}
We present a general method for constructing radar transmit pulse trains and receive filters for which the radar point-spread function in delay and Doppler (radar cross-ambiguity function) is essentially free of range sidelobes inside a Doppler interval around the zero-Doppler axis. The transmit and receive pulse trains are constructed by coordinating the transmission of a pair of Golay complementary waveforms across time according to zeros and ones in a binary sequence $P$. In the receive pulse train filter, each waveform is weighted according to an element from another sequence $Q$. We show that the spectrum of essentially the product of $P$ and $Q$ sequences controls the size of the range sidelobes of the cross-ambiguity function. We annihilate the range sidelobes at low Doppler by designing the $(P,Q)$ pairs such that their products have high-order spectral nulls around zero Doppler. We specify the subspace, along with a basis, for such sequences, thereby providing a general way of constructing $(P,Q)$ pairs. At the same time, the signal-to-noise ratio (SNR) at the receiver output, for a single point target in white noise, depends only on the choice of $Q$. By jointly designing the transmit-receive sequences $(P,Q)$, we can maximize the output SNR subject to achieving a given order of the spectral null. The proposed $(P,Q)$ constructions can also be extended to sequences consisting of more than two complementary waveforms; this is done explicitly for a library of Golay complementary quads. Finally, we extend the construction of $(P,Q)$ pairs to multiple-input-multiple-output (MIMO) radar, by designing transmit-receive pairs of paraunitary waveform matrices whose matrix-valued cross-ambiguity function is essentially free of range sidelobes inside a Doppler interval around the zero-Doppler axis.
\end{abstract}
  

\section{Introduction}

Phase coding \cite{Levanon-book} is a common technique in radar for generating waveforms with impulse-like autocorrelation functions for localizing targets in range. In this technique, a long pulse is phase coded with a unimodular (biphase or polyphase) sequence and the autocorrelation function of the coded waveform is controlled through the autocorrelation function of the unimodular sequence. Examples of sequences that produce good autocorrelation functions are Frank codes \cite{Frank-IT63}, Barker sequences \cite{Barker-book}, and generalized Barker sequences by Golomb and Scholtz \cite{Golomb-IT65}, polyphase sequences by Heimiller \cite{Heimiller-IT61}, and polyphase codes by Chu \cite{Chu-IT72}. It is however impossible to obtain an impulse autocorrelation with a single unimodular sequence. This has led to the idea of using complementary sets of unimodular sequences \cite{Golay-IRE61}--\nocite{Welti-IT60,Turyn-IT63,Taki-IT69}\cite{Tseng-IT72} for phase coding.

Golay complementary sequences (Golay pairs), introduced by Marcel Golay \cite{Golay-IRE61}, have the property that the sum of their autocorrelation functions vanishes at all nonzero lags. Thus, if the waveforms phase coded by Golay complementary sequences (called from here on Golay complementary waveforms) are transmitted separately in time, and their complex ambiguity functions are summed, the result  is essentially an impulse in range along the zero-Doppler axis. In some sense, this makes Golay complementary waveforms ideal for separating point targets in range when the targets have the same Doppler frequency. The concept of complementary sequences has been generalized to multiple complementary codes by Tseng and Liu \cite{Tseng-IT72}, and to multiphase (or polyphase) sequences by Sivaswami \cite{Sivaswami-IT78} and Frank \cite{Frank-IT80}. Properties of complementary sequences, their relationship with other codes, and their applicability in radar have been studied in several articles among which are \cite{Golay-IRE61}--\nocite{Turyn-IT63,Welti-IT60,Taki-IT69,Tseng-IT72,Sivaswami-IT78}\cite{Frank-IT80}. Golay complementary codes have also been used, in conjunctions with space-time coding techniques, to develop waveform matrices with desired correlation and cross-correlation properties \cite{Willett_Golay}.  

In practice, however, a major obstacle exists to adopting Golay complementary sequences for radar; The perfect auto-correlation property of these sequences is extremely sensitive to Doppler shift. Off the zero-Doppler axis the impulse-like response in range is not maintained and the sum of the ambiguity functions of the waveforms has large range sidelobes. In consequence, a weak target  located in range near a strong reflector with a different Doppler frequency may be masked by the range sidelobes of the radar ambiguity function centered at the  delay-Doppler position of the stronger reflector. This is particularly problematic for detecting targets in the presence of clutter, because clutter often occupies a Doppler frequency band around zero. All generalizations of Golay complementary sequences, including multiple complementary sequences and polyphase complementary sequences, suffer from the same problem to varying degrees.

In \cite{Pezeshki-IT08} and \cite{Chi-book}, we showed that, by coordinating the transmission of a pair of Golay complementary waveforms in a pulse train according to the zeros and ones in a binary sequence, called the Prouhet-Thue-Morse (PTM) sequence (see, e.g., \cite{Allouche-SETA98}), we can produce a pulse train whose ambiguity function is essentially free of range sidelobes in a narrow interval around the zero-Doppler axis, and thereby achieving Doppler resilience. We also extended this idea to constructing a PTM sequence of two by two paraunitary waveform matrices that maintain their paraunitary property at modest Doppler shifts. A generalized PTM sequence was subsequently used in \cite{Tang_PTM} for creating Doppler resilience in multiple-input-multiple-output (MIMO) radar transmissions. Our original PTM construction for single channel radar has recently been tested in simulation for Doppler resilience in automotive radar \cite{PTMauto}.   

In this paper, we extend the idea that we introduced in \cite{Pezeshki-IT08} to the {\em joint design} of transmit pulses and receive filters. We develop a systematic approach to sequencing and weighting of Golay complementary waveforms in the transmit pulse train and the receive filter, respectively, to essentially annihilate the range sidelobes of the radar point-spread function inside a Doppler interval around the zero-Doppler axis. We construct the transmit pulse train by coordinating the transmission of a pair of Golay complementary waveforms across consecutive pulse repetition intervals according to a binary sequence $P$ (of $0$s and $1$s). The pulse train used in the receive filter is constructed with the same sequencing of Golay waveforms across time, but each waveform in the pulse train is weighted according to an element of a  nonnegative sequence  $Q$. We call such a transmit-receive pair of pulse trains a $(P,Q)$ pair.    

We show that the size of the range sidelobes of the cross-ambiguity function of a (P,Q) pair, which constitutes the radar point spread function in range and Doppler, is controlled by the spectrum of essentially the product of $P$ and $Q$ sequences in a very precise way. By selecting sequences for which the spectrum of their product has a higher-order null around zero Doppler, we can annihilate the range sidelobe of the cross ambiguity function inside an interval around the zero-Doppler axis. At the same time, the signal-to-noise ratio (SNR) at the receiver output, defined as the ratio of the peak of the squared cross-ambiguity function to the noise power at the receiver output, depends only on the choice of $Q$. By jointly designing the transmit-receive sequences $(P,Q)$, we can maximize the output SNR subject to achieving a given order of the spectral null. 

We discuss two specific $(P,Q)$ designs in detail; namely, the PTM design and the binomial design. In the former, the transmit sequence $P$ is the binary PTM sequence of length $N$ and the weighting sequence $Q$ at the receiver is the all-$1$s sequence. In this case, the output SNR in white noise is maximum, as the receiver filter is in fact a matched filter, but  the order of the spectral null is only  logarithmic in the length $N$ of the transmit pulse train. This design was originally introduced in \cite{Pezeshki-IT08}. We present it in this paper for comparison with other designs. In the binomial design, on the other hand,  $P$ is the alternating binary sequence of length $N$ and $Q$ is the sequence of binomial coefficients in the binomial expansion $(1+z)^{N-1}$. In this case, the order of the spectral null is $N-2$; the largest that can be achieved with a pulse train of length $N$, but  this comes at the expense of SNR.

For general designs, we derive a necessary and sufficient condition for achieving an $M$th-order spectral null with length-$N$ $(P,Q)$ pairs. More specifically, we determine the subspace of length $N$ sequences that have spectral nulls of order $M$ at zero, and identify a basis for this subspace. We also present an alternative characterization of such sequences via a null space condition. We then formulate the problem of designing $(P,Q)$ pairs as one of maximizing the output SNR subject to the subspace constraint for achieving a null of a given order. The signs of the elements of the solution determine the sequence $P$ and the moduli of its elements determine the sequence $Q$. 

We also propose a systematic extension of the above $(P,Q)$ pulse trains to waveform libraries with more than two complementary waveforms. Here, we construct $2^m$-ary sequences $P_m$ for coordinating the transmission of $2^m$ complementary waveforms and the corresponding $Q_m$ sequences for weighting them in the receive filter to suppressing range sidelobes. We present an explicit example of such designs for Golay complementary quads.  

Finally, we extend the construction of $(P,Q)$ pairs to multiple-input-multiple-output (MIMO) radar, by designing transmit-receive pairs of paraunitary waveform matrices whose matrix-valued cross-ambiguity function is essentially free of range sidelobes inside a Doppler interval around the zero-Doppler axis.

\begin{remark}
We note that this paper concerns the construction of transmit-receive pairs of complementary waveforms that exhibit Doppler resilience, and not the design of unimodular sequences with Doppler resilience,  which has been studied by several authors. For example, in \cite{Sivaswami-AES82} and \cite{Bell-IT98} a class of near complementary  codes, called subcomplementary codes, that exhibits some tolerance to Doppler shift has been introduced. The term near complementary means that the sum of the autocorrelations of the sequences is not an impulse and has modest sidelobes in delay. Also a large body of work exists concerning the design of single polyphase sequences that have Doppler tolerance. A few examples are Frank codes \cite{Frank-IT63}, $P1$, $P2$, $P3$, and $P4$ sequences \cite{Lewis-AES83}, $PX$ sequences \cite{Rapajik-CommLett98}, and $P(n,k)$ sequences \cite{Felhauer-ELett92},\cite{Felhauer-AES94}. The design of Doppler tolerant polyphase sequences has also been considered for MIMO radar \cite{Khan-SPL06} and for orthogonal netted radar \cite{Deng-SP04}.
\end{remark}

\section{Complementary Waveforms and Their Ambiguities}

Let $\Omega(t)$ denote a baseband pulse shape with duration limited to
a chip interval $T_c$ and unit energy:
\begin{equation}
\int_{-T_c/2}^{T_c/2}|\Omega(t)|^2 dt = 1.
\end{equation}
The ambiguity function $\chi_{\Omega}(\tau,\nu)$ of $\Omega(t)$ is 
\begin{equation}
\chi_{\Omega}(\tau,\nu)=\int_{-T_c}^{T_c}\Omega(t)\Omega^{*}(t-\tau)e^{-j\nu t}dt,
\end{equation}
where $\tau$ and $\nu$ are delay and Doppler frequency variables, respectively, and $^*$ denotes complex conjugation.
 
A baseband waveform constructed by phase coding translates of $\Omega(t)$ with a length $L$ unimodular sequence $w[n]$ can be be expressed as     
\begin{equation}
  \label{eq:sum_of_chips}
  w(t)=\sum\limits_{\ell=0}^{L-1}w[\ell]\Omega(t-\ell T_c).
\end{equation}
The energy of $w(t)$ is 
\begin{align}\label{eq:Ew}
E_w &=\int_{\bR} |w(t)|^2 dt \nonumber\\
&=\left(\sum\limits_{\ell=0}^{L-1} w[\ell]^2\right)\int\limits_{-T_c/2}^{T_c/2} |\Omega(t)|^2 dt\nonumber\\
&=L.
\end{align}
The ambiguity $\chi_{w}(\tau,\nu)$ of $w(t)$ at delay-Doppler coordinates $(\tau,\nu)$ is
\begin{multline}\label{eq:chi_x}
\chi_{w}(\tau,\nu)=\int_{-\infty}^{\infty}w(t)w(t-\tau)^*e^{-j
\nu t} dt \\
%=\sum_{\ell=0}^{L-1}\sum_{k=-(L-1)}^{L-1}w[\ell]w[\ell-k]^*\\ \times \int_{\bR}\Omega(t-\ell T_c)\Omega(t-(\ell-k)T_c-\tau)^*e^{-j\nu t}dt \\
=\sum_{k=-(L-1)}^{L-1}A_w(k,\nu T_c)\chi_{\Omega}(\tau-kT_c,\nu),
\end{multline}
where 
\begin{equation}
A_w(k,\nu T_c)=\sum_{\ell=0}^{L-1}w[\ell]w[\ell-k]^*e^{-j\nu\ell T_c}
\end{equation}
for $k=-L-1, -L, \ldots, L-1$.

\textit{Definition 1: Golay Complementary Pair 
  \cite{Golay-IRE61}}. Two length $L$ unimodular sequences of complex
numbers $x[\ell]$ and $y[\ell]$ form a \emph{Golay
complementary pair} if, for all lags
$k=-(L-1),-(L-2),\ldots,(L-1)$, their summed  autocorrelation functions
satisfies
\begin{equation}
C_x[k]+C_y[k]=2L\delta[k],
\end{equation}
where $C_x[k]$ and $C_y[k]$ are the autocorrelations of $x[\ell]$ and $y[\ell]$ at lag $k$, respectively, and $\delta[k]$ is the Kronecker
delta function. From here on, we may drop the discrete time index $\ell$ from $x[\ell]$ and
$y[\ell]$ and simply use $x$ and $y$ when appropriate. 

\begin{remark}\label{rm:golayconst}
The sequences $[1,1]$ and $[1,-1]$ are Golay complementary. Golay complementary sequences of length $2^{m+1}$ can be constructed recursively from Golay complementary sequences of length $2^{m}$ for $m\ge 1$. Let $\mathbf{G}_{2^m}$ be a $2^m \times 2^m$ matrix, in which every pair of rows is Golay complementary. Partition $\mathbf{G}_{2^m}$ as
\begin{equation}
\mathbf{G}_{2^m}=\begin{bmatrix} \hfill \mathbf{F}_{1} \\ \hfill \mathbf{F}_{2} \end{bmatrix},
\end{equation} 
where $\mathbf{F}_{1}$ and $\mathbf{F}_{2}$ are $2^{m-1}\times 2^m$ matrices. Then, $\mathbf{G}_{2^{m+1}}$ can be constructed from $\mathbf{G}_{2^m}$ as \cite{Golay-IRE61}
\begin{equation}
\mathbf{G}_{2^{m+1}}=\begin{bmatrix} \hfill\mathbf{F}_{1} & \hfill \mathbf{F}_{2} \\ \hfill\mathbf{F}_{1} & \hfill-\mathbf{F}_{2}\\ \hfill \mathbf{F}_{2} & \hfill \mathbf{F}_{1} \\ \hfill\mathbf{F}_{2} & \hfill-\mathbf{F}_{1}\end{bmatrix}.
\end{equation}
The following example shows the construction of Golay complementary sequences of length four from those of length two:  
\begin{equation}
\begin{bmatrix} \hfill 1 & \hfill 1 \\ \hfill 1 & \hfill -1 \end{bmatrix}\longrightarrow \begin{bmatrix} \hfill 1 & \hfill 1 & \hfill 1 & \hfill -1 \\ \hfill 1 & \hfill 1 & \hfill -1 & \hfill 1 \\ \hfill 1 & \hfill -1 & \hfill 1 & \hfill 1 \\ \hfill 1 & \hfill -1 & \hfill -1 & \hfill -1 \end{bmatrix}.
\end{equation}
\end{remark}

A pair of baseband waveforms $x(t)$ and $y(t)$, phase
coded by length-$L$ Golay complementary sequences $x$ and
$y$: that is, 
\[
x(t)=\sum\limits_{\ell=0}^{L-1}x[\ell]\Omega(t-\ell T_c)
\] 
and 
\[
y(t)=\sum\limits_{\ell=0}^{L-1}y[\ell]\Omega(t-\ell T_c), 
\]
individually, have ambiguity functions $\chi_x$ and $\chi_y$, as in \eqref{eq:chi_x}, with $x$ and $y$ replacing $w$. 

Separating the transmissions of $x(t)$ and $y(t)$ in time by a pulse repetition interval (PRI) of  
$T$ seconds results in the radar waveform $z(t)=x(t)+y(t-T)$ with the ambiguity function  
\begin{equation}\label{eq:Xs}
\chi_{z}(\tau,\nu)=\chi_{x}(\tau,\nu)+e^{-j \nu
T}\chi_{y}(\tau,\nu).
\end{equation}
\begin{remark}\label{rem:001}
The ambiguity function of $z(t)$ has two range aliases (cross terms)  offset from the zero-delay
axis by $\pm T$. In this paper, we ignore  the range aliasing effects and refer to $\chi_{z}(\tau,\nu)$ as the ambiguity
function of $z(t)$. Range aliasing effects can be accounted for using standard techniques (for instance,  see \cite{Levanon-book}) and will not be further discussed.
\end{remark}

As the duration $LT_c$ of the waveforms $x(t)$ and $y(t)$ is typically much shorter than the
PRI duration $T$, the Doppler shift over $LT_c$, i.e., $\nu LT_c$ is negligible compared to the Doppler shift over a PRI, i.e., $\nu T$, for 
practically feasible targets . \emph{A fortiori} this applies to the Doppler shift over an 
individual chip (a single translate of $\Omega$). As a result, we can approximate $\chi_{z}(\tau,\nu)$ as 
\begin{multline}
\chi_{z}(\tau,\nu) = \sum\limits_{k=-(L-1)}^{L-1}A_x(k,0)\chi_{\Omega}(\tau-kT_c,0)\\ 
 +  e^{j\nu T} \sum\limits_{k=-(L-1)}^{L-1}A_y(k,0)\chi_{\Omega}(\tau-kT_c,0).
\end{multline} 
Noting that $A_x(k,0)=C_x[k]$, $A_y(k,0)=C_y[k]$, and defining $C_\Omega(\tau)$ as the autocorrelation function of $\Omega_t$, that is, $C_\Omega(\tau)= \chi_{\Omega,0}(\tau)$, we have 
\begin{multline}\label{chi_s}
\chi_z(\tau,\nu)=\\
\sum\limits_{k=-(L-1)}^{L-1}\left(C_x[k]+e^{-j\nu
T}C_y[k]\right)C_{\Omega}(\tau-kT_c).
\end{multline}

Along the zero-Doppler axis ($\nu=0$), the ambiguity function
$\chi_z(\tau,\nu)$ reduces to
\begin{equation}
\chi_z(\tau,0)=2L C_\Omega(\tau),
\end{equation}
by complementarity of the Golay pair $x$ and $y$.  We observe that
the ambiguity function $\chi_z(\tau,\nu)$ is ``free'' of range
sidelobes along the zero-Doppler axis. However, it is known (see, e.g., \cite{Levanon-book}) that off
the zero-Doppler axis the ambiguity function has large sidelobes in
delay (range).

\begin{remark}
The shape of $\chi_z(\tau,0)$ in delay depends on the shape of the autocorrelation function
  $C_{\Omega}(\tau)$ of $\Omega(t)$. The Golay
  complementary property eliminates range sidelobes caused by
  replicas of $C_\Omega(\tau)$ at nonzero lags.
\end{remark}

\begin{remark}\label{rem:111}
One might think that separating Golay complementary waveforms in frequency (transmitting them over non-interfering frequency bands) would also result in range sidelobe  cancellation along the zero-Doppler axis. But this is not the case, as the
presence of delay-dependent phase terms impairs the complementary
property of the waveforms. Searle and Howard
\cite{Searle-WDDC07, Searle-RSC07, Searle09} have introduced 
modified Golay pairs for OFDM channel models. These modified Golay pairs
are complementary in the sense that the sum of their squared autocorrelation
functions forms an impulse in range.
\end{remark}
 
\section{$(P,Q)$ Pulse Trains of Complementary Waveforms}\label{sc:PQbinary}

We now consider the effect of transmitting a longer sequence of Golay
pairs as a pulse train over multiple PRIs. 

\textit{Definition 2: } Let $P=\{p_n\}_{n=0}^{N-1}$ be a
binary sequence of length $N$. The \textit{$P$-pulse train}
$z_{P}(t)$ is defined as 
\begin{equation}
z_{P}(t)=\sum\limits_{n=0}^{N-1}p_nx(t-nT)+\overline{p}_ny(t-nT),
\end{equation}
where $\overline{p}_n=1-p_n$ is the complement of $p_n$. The $n$th
pulse  in
$z_{P}(t)$ is $x(t)$ if $p_n=1$ and is $y(t)$ if
$p_n=0$, and  consecutive pulses  are
separated in time by a PRI $T$.  

\textit{Definition 3: } Let $Q=\{q_n\}_{n=0}^{N-1}$ be a
discrete real nonnegative sequence ($q_n\ge 0$) of length $N$. The \textit{$Q$-pulse train} $z_{Q}(t)$ is defined as 
\begin{equation}
z_{Q}(t)=\sum_{n=0}^{N-1}q_{n}\bigl[p_nx(t-nT)+\overline{p}_ny(t-nT)\bigr].
\end{equation}
The $n$th pulse  in $z_{Q}(t)$ is obtained by multiplying
the $n$th pulse in  $z_{P}(t)$ by $q_n$.  

We refer to the
pair $(z_{P}, z_{Q})$ as the $(P,Q)$-\emph{transmit-receive pair} or just \emph{$(P,Q)$-pair}. Transmitting $z_{P}(t)$ and filtering the return by
(correlation with) $z_{Q}(t)$ results in a  point-spread function
(in delay and Doppler) that is given by the cross-ambiguity function between $z_{P}(t)$
and $z_{Q}(t)$:  
\begin{equation}\begin{split}\label{cross_amb}
\chi_{{PQ}}(\tau, \nu)
&=\int_{\bR}z_{P}(t)z_{Q}(t-\tau)^*e^{-j\nu t}dt \\
&=\sum_{n=0}^{N-1}q_ne^{-j\nu nT}\bigl[p_n\chi_{x}(\tau,\nu)
+\overline{p}_n\chi_{y}(\tau,\nu)\bigr],
\end{split}\end{equation}
where, as in  \eqref{eq:Xs}, range aliases at offset $\pm nT$,
$n=1,2,...,N-1$ are ignored. 

\begin{remark}
When $q_n=1$ for $n=0,...,N-1$, the receiver is a matched filter that matches to the transmitted pulse train
$z_{P}(t)$ and (\ref{cross_amb}) reduces to the ambiguity
function of $z_{P}(t)$. The joint design of $P$ and $Q$ provides  more flexibility in tailoring the shape of the radar cross ambiguity function, as will be demonstrated in the next section.
\end{remark}

As in \eqref{chi_s}, $\chi_{{PQ}}(\tau, \nu)$ is well approximated by
\begin{multline}\label{eq:amb1}
\chi_{{PQ}}(\tau, \nu)
=\sum_{n=0}^{N-1}q_ne^{-j\nu nT}\\ \times
\sum_{k=-(L-1)}^{L-1}\bigl[p_nC_x[k]+\overline{p}_nC_y[k]\bigr]C_{\Omega}(\tau-kT_c).
\end{multline}
A key fact to take from \eqref{eq:amb1} is that it is a linear
combination of translates of $C_\Omega$, and of course that this
function has support twice the chip-length about the
origin. Accordingly, it is convenient to leave this aside and consider
\begin{multline}\
\chi_{PQ}\left(k,\theta\right)\\
=\sum_{k=-(L-1)}^{L-1}\bigl[p_nC_x[k]+\overline{p}_nC_y[k]\bigr]\sum_{n=0}^{N-1}q_ne^{-j\nu T}, 
\end{multline}
which, after some simple algebraic manipulations, we can write as  
\setlength\arraycolsep{0pt}
\begin{multline}\label{eq:ptmamb}
\chi_{PQ}\left(k,\theta\right)=\frac{1}{2}\bigl[C_x[k]+C_y[k]\bigr]\sum\limits_{n=0}^{N-1}q_{n}e^{jn\theta}
\\ -\frac{1}{2}\bigl[C_x[k]-C_y[k]\bigr]\sum\limits_{n=0}^{N-1}(-1)^{p_{n}}q_{n}
e^{jn\theta},
\end{multline}
where $\theta=\nu T$ is the relative Doppler shift over a PRI $T$. 

The form in \eqref{eq:ptmamb} is particularly convenient for studying range sidelobes. Since
$(x,y)$ is a Golay pair, $C_x[k]+C_y[k]=2L\delta[k]$
and the first term on the right-hand-side of \eqref{eq:ptmamb} is free
of range sidelobes. The second term, then, provides the obstacle to a
perfect cross-ambiguity, and is to be controlled by choice of $P$ and
$Q$.

The main question to be addressed is: \emph{Can sequences $P$ and $Q$ be designed
so that the cross-ambiguity $\chi_{PQ}\left(k,\theta\right)$ is essentially a Kronecker delta in
delay, at least for some range of Doppler frequencies?}

\begin{remark}
In our previous work  \cite{Pezeshki-IT08}, we looked only at designing the sequence $P$ and simply took $Q$ to be an all one sequence. In the present paper, we show that the joint design of $P$ and $Q$  considerably enriches our choices for suppressing range sidelobes around zero Doppler.   
\end{remark}

\section{Range Sidelobe Suppression}\label{sc:RSS}

The  spectrum of the sequence $r_n=(-1)^{p_n}q_n$, $n=0,\ldots,N-1$,
\begin{equation}\label{eq:Sp}
S_{PQ}(\theta)=\sum_{n=0}^{N-1}r_ne^{jn\theta},
\end{equation}
is the key component of the term with range sidelobes in \eqref{eq:ptmamb}. By selecting sequences $P$ and $Q$ such that $S_{PQ}(\theta)$ has a high order null at $\theta=0$, the range sidelobes in a Doppler interval around zero can be suppressed. 

We begin with two examples of the kind of effects that $(P,Q)$ pairs can achieve. We first present the \emph{Prouhet-Thue-Morse (PTM)} design from our previous work \cite{Pezeshki-IT08,Chi-book}. In this design, the binary sequence $P$ is a PTM sequence (see, e.g., \cite{Allouche-SETA98}) of length $N$, where $N$ is a power of $2$. The sequence $Q$ is an all one sequence, meaning that the receiver is a matched filter. This design achieves a spectral null of order $\log_2 N-1$ around $\theta=0$ in $S_{PQ}(\theta)$. In the second example, the {\em binomial} design, $P$ is an alternating binary sequence of length $N$, meaning that the transmitter alternates between the two Golay complimentary waveforms $x(t)$ and $y(t)$ in consecutive PRIs. The sequence $Q$ is the sequence of binomial coefficients in the binomial expansion $(1+z)^{N-1}$. This design achieves a spectral null of order $N-2$ around $\theta=0$ in $S_{PQ}(\theta)$, which is the highest order null achievable with $(P,Q)$ pairs. 

After these two examples, we derive a general way of constructing $(P,Q)$ pairs of length $N$ for achieving a spectral null of order $M\le N-2$. We also lower bound the peak-to-peak-sidelobe ratios of the cross-ambiguity function associated with such $(P,Q)$ pairs.  

Later, in Section \ref{sc:snr}, we derive an expression for the output SNR of $(P,Q)$ pairs, for a single point target in white noise, and discuss the construction of maximum SNR $(P,Q)$ pairs that achieve a given order of spectral null. The PTM design has maximum SNR, because it uses a matched filter at the receiver. The binomial design uses a different receiver which enables us to produce the largest order of null possible, at the expense of SNR.      
  
 \subsection{PTM vs. Binomial Design}
The following theorem was proved in \cite{Pezeshki-IT08}:
\begin{theorem}[PTM Design]
  \label{thm:1} Let
$P=\{p_n\}_{n=0}^{N-1}$ be the length $N=2^{M+1}$, $M\ge 1$, Prouhet-Thue-Morse
(PTM) sequence, defined
recursively by $p_{2k}=p_{k}$ and $p_{2k+1}=1-p_{k}$ for all $k\geq
0$, with $p_0=0$, and let $Q=\{q_n\}_{n=0}^{N-1}$ be the 
sequence $1$s of length $N=2^{M+1}$. Then, $S_{{P,Q}}(\theta)$ has an
$M$th-order null at $\theta=0$.
\end{theorem}

\begin{example}
 The PTM sequence of length $N=8$ is $P=\{p_k\}_{k=0}^{7} \ = \ 0
\ 1 \ 1 \ 0\ 1\ 0\ 0\ 1$, and  the corresponding $P$-pulse train is
\begin{multline*}
z_{P}(t)= x(t)+y(t-T)+y(t-2T)+x(t-3T)\\
+y(t-4T)+x(t-5T)
+x(t-6T)+y(t-7T).
\end{multline*}
The receive filter pulse train, $z_{Q}(t)$, is chosen to be the same
as the $P$-pulse train. The order of the null of $S_{P,Q}(\theta)$ is $M=\left(\log_{2}N\right)-1=2$.
\end{example}

\begin{remark}
The first $M$ moments of $S_{PQ}(\theta)$ about $\theta=0$ are 
\[
S_{PQ}^{(m)}(0)=\sum\limits_{n=0}^{N-1} (-1)^{p_n} n^m, \quad m=0,1,\ldots, M.
\]
Forcing these moments to vanish requires balancing out the sum of the powers of integers that get positive signs with the sum of those that get negative signs. This is where the PTM sequence comes in. Let $\mathbb{S}=\{0,1,\ldots,N-1\}$ be 
the set of all integers between $0$ and $N-1$. The Prouhet problem is the following. Given $M$, is it possible to partition $\mathbb{S}$ into two disjoint subsets $\mathbb{S}_0$ and $\mathbb{S}_1$ such that 
\[
\mathop{\sum}\limits_{k\in \mathbb{S}_0}k^m=\mathop{\sum}\limits_{l\in \mathbb{S}_1}l^m
\]
for all $0\le m\le M$? Prouhet proved that this is possible when $N=2^{M+1}$ and that the partitions are identified by the PTM sequence. The set $\mathbb{S}_0$ consists of all integers $n\in \mathbb{S}$ where the PTM sequence $p_n$ is zero, and the set $\mathbb{S}_1$ consists of all integers $n\in \mathbb{S}$ where the PTM sequence $p_n$ is one. The reader is referred to \cite{Allouche-SETA98} for a review of problems and results related to the PTM sequence.
\end{remark}

The PTM design for radar transmissions was originally introduced in
\cite{Pezeshki-IT08,Chi-book} in the context of
Doppler resilient waveforms. Here, we further investigate
this design and  compare  it with the Binomial design, to be
described next. 

\begin{theorem}[Binomial Design]
  \label{thm:2}
  Let $P=\{p_n\}_{n=0}^{N-1}$ be the length $N=M+2$, $M\ge 1$, alternating
  sequence, where $p_{2k}=1$ and $p_{2k+1}=0$ for all $k\ge 0$, and
  let $Q=\{q_n\}_{n=0}^{N-1}$ be the length $N=M+2$ binomial sequence
  $\{q_{n}\}_{n=0}^{N-1}=\{\binom{N-1}{n}\}_{n=0}^{N-1}$. Then,
  $S_{{PQ}}(\theta)$ has an $M$th order null at $\theta=0$.
\end{theorem}

\begin{proof}
The spectrum $S_{P,Q}(\theta)$ for the alternating sequence $P$ and binomial sequence $Q$ is
\begin{align}\label{sidelobes}
S_{P,Q}(\theta)&=\sum_{n=0}^{N-1}(-1)^{n}\binom{N-1}{n}e^{jn\theta}\nonumber\\ 
&=(1-e^{j\theta})^{N-1}.
\end{align}
It is straightforward to show that $S_{P,Q}^{(m)}(0)=0$ for
$m=0,1,...,N-2$. 
\end{proof}

\begin{example} For $N=8$, the $(P,Q)$ pair is 
\begin{align*}
z_{P}(t)&=x(t)+y(t-T)+x(t-2T)\\
&+y(t-3T)+x(t-4T)+y(t-5T)\\ 
&+x(t-6T)+y(t-7T),\\
z_{Q}(t)&=q_0x(t)+q_1y(t-T)+q_2x(t-2T)\\
&+q_3y(t-3T)+q_4x(t-4T)+q_5y(t-5T)\\
&+q_6x(t-6T)+q_7y(t-7T),
\end{align*}
where $q_n=\binom{7}{n}$, $n=0,1,...,7$. The order of the spectral null for sidelobe suppression is $M=N-2=6$.
\end{example}

\subsection{General ${(P,Q)}$ pair design}\label{sec:2}

We characterize, here, those $(P,Q)$ sequences of length $N > M + 1$
that have  an $M$-th order spectral null. 

Let  $V_{N}$ be the real
vector space of (trigonometric)
polynomials of degree at most $N-1$, regarded as functions on
$[0,2\pi]$: 
\begin{multline}
V_{N}=\{f\colon f(\theta) = \sum_{n=0}^{N-1} r_n e^{j n
  \theta}=\sum_{n=0}^{N-1} r_{n}z^n, \\ \quad z=e^{jn\theta},\ r_n\in \mathbb R,\ \theta\in [0,2\pi]\},
\end{multline}
With some abuse of notation, we  will freely switch between regarding $f$
as a function of $\theta$ and as a function of $z$.

Each $f\in V_N$ corresponds to a $(P,Q)$ sequence of length $N$ through $r_n=(-1)^{p_n}q_n$, $n = 0,...,N-1$, or equivalently
\begin{equation}\label{eq:pn}
p_n = \frac{1}{2}(1-\sgn(r_n)), q_n = |r_n|,\ n = 0,...,N-1.
\end{equation}
Evidently,  the functions $g_n$, $n=0,\ldots,N-1$, given by
\begin{equation}\label{eq:gntheta}
g_n(\theta) = (1-e^{j\theta})^n=(1-z)^{n},
\end{equation}
also form a basis for $V_N$ (besides $\{e^{jn\theta}\}_{n=0}^{N-1}$); in fact, by the Binomial Theorem, 
\begin{equation}\label{eq:eintheta}
e^{j n \theta} = \sum_{k=0}^{N-1} (-1)^{k}\binom{n}{k}\, g_k(\theta), \quad n=0,\cdots, N-1,
\end{equation}
where we use the convention that $\binom{n}{k}=0$ if $k>n$.
We write
\begin{equation}\label{eq:specnullVN}
T_{M}=\{f\in V_{N}\colon \frac{d^{m}f}{d\theta^{m}} (0) = 0, \ \text{for $m=0,\cdots, M$}\} 
\end{equation}
for the (manifestly) linear subspace of  $V_{N}$
consisting of those functions that  have an $M$-th order spectral
null. 
Since $\frac{1}{j}\frac{d}{d\theta}=z\frac{d}{dz}$ and
\begin{equation}
  \label{eq:1}
  \Bigl(z\frac{d}{dz}\Bigr)^{k}=z^{k}\frac{d^{k}}{dz^{k}}+
  \text{ lower order terms in  $\frac{d}{dz}$},
\end{equation}
$f$ (\emph{qua} polynomial in $z$) has a spectral null of order $M$ if and only if
\begin{equation}\label{eq:specnullVN1}
\frac{d^{m}f}{dz^{m}}(1) = 0, \ \text{for $m=0,\cdots, M$}.
\end{equation}
An invocation of Taylor's Theorem, yields the following result.  
\begin{theorem}\label{thm:2a}
	The subspace $T_M \subset V_N$ of functions having an $M$-th order spectral null
	is spanned by $\{g_n |\; n=M+1, \cdots, N-1\}$.
\end{theorem}

\begin{remark}  $T_{N-2}$ has dimension 1 and consists of real
  multiples of $g_{N-1}$; that is,  
  given a fixed length of
pulse train $N$  the binomial design is the only choice of ${(P,Q)}$
pair (up to a scale factor in $Q$) to achieve the highest order null.
\end{remark}

By Theorem \ref{thm:2a}, any $S_{P,Q}(\theta)=\sum\limits_{n=0}^{N-1}r_n e^{jn\theta}$ with an spectral null of order $M\le N-2$ can be expressed as  
\begin{align}\label{eq:SPQa}
S_{P,Q}(\theta)&=\sum\limits_{m=0}^{N-M-2} a_{m} g_{m+M+1}(\theta)\nonumber\\
&=\sum\limits_{m=0}^{N-M-2} a_{m} (1-e^{j\theta})^{m+M+1}
\end{align}
for some $\mathbf{a}=[a_{0},\ldots,a_{N-M-2}]^T\in \mathbb{R}^{N-M-1}$ (with $\mathbf{a}\neq \mathbf{0}$). By the binomial theorem, we can write 
\begin{align}\label{eq:rcomb}
S_{P,Q}(\theta)&=\sum\limits_{m=0}^{N-M-2} a_m\sum\limits_{n=0}^{m} (-1)^{n}\binom{m+M+1}{n} e^{jn\theta}\nonumber\\
&=\sum\limits_{n=0}^{N-1}\sum\limits_{m=0}^{N-M-2}a_{m}(-1)^n \binom{m+M+1}{n}e^{jn\theta},
\end{align}
where we have used the convention that $\binom{m+M+1}{n}=0$ if $n>m+M+1$. Thus, we obtain  
\begin{equation}
r_n=\left(\sum\limits_{m=0}^{N-M-2}a_{m}(-1)^n \binom{m+M+1}{n}\right)
\end{equation}
or in vector form 
\begin{equation}\label{eq:rBM}
\mathbf{r} = \mathbf{B}_M \mathbf{a} , \ \mathbf{a}\neq \mathbf{0},
\end{equation}
where  $\mathbf{r}=(r_0, \cdots, a_{N-1})^T\in \mathbb R^{N-1}$ and $\mathbf{B}_M$ is an $N\times (N-M-1)$ matrix whose $(n,m)$th entry is 
\begin{equation}\label{eq:BM}
\left(\mathbf{B}_M\right)_{m,n}=(-1)^{n} \binom{m+M+1}{n} ,  
\end{equation}
for $n=0,\ldots, N-1$ and $m=0,\ldots, N-M-2$. In other words, $S_{P,Q}(\theta)$ has an $M$th order null at $\theta=0$ if and only if $\mathbf{r}$ is in the space spanned by the columns of $\mathbf{B}_M$. For each $\mathbf{r}$ constructed in this fashion, we can obtain the corresponding $P$ and $Q$ sequences as in \eqref{eq:pn}, with the convention that $\sgn(0)=1$.  
 
\begin{remark}\label{rm:null}
As an alternative way of characterizing $r_n=(-1)^{p_n}q_n$, $n=0,\ldots,N-1$, we note that the vector $\mathbf{r}\neq 0$ lies in the null
space of an $(M+1)\times N$ integer Vandermonde matrix $\mathbf{V}_M$, whose
$(m,n)$th element is $n^{m}$, $m=0,\ldots,M$ and $n=0,\ldots,N-1$, that is,
\begin{equation}\label{eq:null}
\mathbf{V}_M\mathbf{r}=0. 
\end{equation}
To see this, consider 
the Taylor expansion of $S_{P,Q}(\theta)$
around $\theta=0$:
\begin{equation}
S_{PQ}(\theta)=\sum_{m=0}^{\infty}S_{P,Q}^{(m)}(0)\frac{\theta^m}{m!},
\end{equation}
where  
\begin{equation}
S_{PQ}^{(m)}(0)=\sum_{n=0}^{N-1}n^m (-1)^{p_n}q_n=\sum_{n=0}^{N-1}n^m r_n
\end{equation}
is the $m$-th order derivative of
$S_{P,Q}(\theta)$ at $\theta=0$. We wish to have 
\begin{equation}\label{eq:null_space}
S_{P,Q}^{(m)}(0)=0,\quad m = 0,1,...,M,
\end{equation}
or equivalently,
\begin{equation}\label{eq:null_space1}
\sum_{n=0}^{N-1}n^mr_n = 0,\quad m=0,1,...,M.
\end{equation}
Writing the above condition in matrix form gives the stated null space result. 
\end{remark}

\section{Signal-to-Noise Ratio}\label{sc:snr}

Suppose that the noise process at the receiver input is white with power $N_0$. Then the noise power at the receiver output is
\begin{align}
\eta&=N_0\int_{\bR} |z_{Q}(t)|^2dt  \nonumber \\ 
& =N_0\sum\limits_{n=0}^{N-1}\sum\limits_{m=0}^{N-1}q_n q_m\int_{\bR}^{}\bigl[p_nx(t-nT)+\overline{p}_ny(t-nT)\bigr] \nonumber
\\ &  \quad\quad\quad\quad\quad \times\bigl[p_mx(t-mT)+\overline{p}_my(t-mT)\bigr]^{*} dt \nonumber\\
&=N_0\sum\limits_{n=0}^{N-1} q_n^2 \int_{\bR} |p_nx(t-nT)+\overline{p}_ny(t-nT)|^2 dt,\label{eq:eta1}
\end{align}
where the last equality follows because the durations of $x(t)$ and $y(t)$ are $LT_c<<T$, and therefore the cross-terms are zero. Noting that $p_n$ and $\overline{p}_n$ are binary complements of each other, each term of the summation in the last line of \eqref{eq:eta1} is either the energy of $x(t)$ or the energy of $y(t)$. The energies of $x(t)$ and $y(t)$ are $E_x=E_y=L$ (see \eqref{eq:Ew}). Thus, 
\begin{equation}\label{eq:eta2}
\eta=N_0 L \| \mathbf{q}\|^2=N_0 L \| \mathbf{r}\|^2, 
\end{equation}   
where $\mathbf{q}=[q_0,...,q_{N-1}]^T$ and $\mathbf{r}=[r_0,\ldots,r_{N-1}]^T=[(-1)^{p_0}q_0,\ldots,(-1)^{p_{N-1}}q_{N-1}]^T$. 

For a single point target, the SNR at the receiver output is
\begin{equation}
\rho = \frac{\sigma_b^2|\chi_{P,Q}(0,0)|^2}{\eta}
=\frac{L\sigma_b^2}{N_0}\frac{\|\mathbf{q}\|_1^2}{\|\mathbf{q}\|_2^2}
=\frac{L\sigma_b^2}{N_0}\frac{\|\mathbf{r}\|_1^2}{\|\mathbf{r}\|_2^2},
\end{equation}
where $\sigma_b^2$ is the power of the target.

The output SNR $\rho$ is maximized when $\mathbf{q}=\alpha\mathbf{1}$ for
some  scalar $\alpha>0$, so that $z_{Q}(t)=\alpha z_{P}(t)$ is the usual matched filter. Any sequence $Q$ other
than this results in a reduction in output SNR. On the other hand, as it was shown in Section \ref{sc:RSS}, the
extra degrees of freedom provided by a more general $Q$ can be used to
create a spectral null of higher order, through the joint design of
$P$ and $Q$, than is achievable by designing $P$ alone.

From \eqref{eq:rBM}, or equivalently \eqref{eq:null}, we see that infinitely many designs $\mathbf{r}$ can achieve a null of order $M<N-2$ in $S_{PQ}(\theta)$. But these designs are different in terms of SNR. The design with the largest SNR is the solution $\mathbf{r}_{M}$ to  
\begin{align}
\mathrm{maximize\ } & \ \ \frac{\|\mathbf{r}\|_{1}^{2}}{\|\mathbf{r}\|_{2}^{2}}\nonumber\\
\mathrm{subject \ to} & \ \ \mathbf{V}_M\mathbf{r}=\mathbf{0},\nonumber\\
& \ \ \mathbf{r}\neq \mathbf{0}.
\end{align}
where $\mathbf{V}_M$ is defined in Remark \ref{rm:null}. We refer to $\mathbf{r}_{M}$ or the corresponding $(P,Q)$ pair as the Max-SNR design for an spectral null order of order $M$. Equivalently, the Max-SNR design can be obtained by solving 
\begin{align}\label{eq:maxsnr0}
\mathrm{minimize\ } & \ \ \|\mathbf{r}\|_{2}^{2}\nonumber\\
\mathrm{subject \ to} & \ \ \|\mathbf{r}\|_{1}=1,\nonumber\\ 
& \ \ \mathbf{V}_M\mathbf{r}=\mathbf{0}.
%& \ \ \mathbf{r}\neq \mathbf{0}.
\end{align}
Given any element of $\mathbf{r}$, say $r_n$, we can always find unique numbers $s_n,t_n\ge 0$, such that $r_n=s_n-t_n$ and $|r_n|=s_n+t_n$. Let $\mathbf{s}=[s_0,\ldots,s_{N-1}]^T$ and $\mathbf{t}=[t_0,\ldots,t_{N-1}]^T$ be vectors of such numbers for the elements of $\mathbf{r}$. Then, we can write the optimization problem in \eqref{eq:maxsnr0} as
\begin{align}\label{eq:maxsnr}
\mathrm{minimize\ } & \ \ \begin{bmatrix}\mathbf{s}^T \ & \  \mathbf{t}^T\end{bmatrix}\begin{bmatrix}\hfill 1 \ & \ \hfill -1\\ \hfill -1 \ & \ \hfill 1\end{bmatrix}\begin{bmatrix}\mathbf{s} \\ \mathbf{t}\end{bmatrix}\nonumber\\  
\mathrm{subject \ to} & \ \  \begin{bmatrix} \mathbf{1}^T \ & \ \mathbf{1}^T\\ \mathbf{V}_M \ & \ \mathbf{V}_M\end{bmatrix}\begin{bmatrix} \mathbf{s} \\ \mathbf{t}\end{bmatrix}=\begin{bmatrix} 1 \hfill\\ \mathbf{0}\end{bmatrix},\nonumber\\
& \ \ \begin{bmatrix} \mathbf{s} \\ \mathbf{t} \end{bmatrix}\ge \mathbf{0}.
\end{align}
where $\mathbf{1}$ and $\mathbf{0}$, respectively, denote all one and all zero vectors of appropriate sizes, and $\ge$ in the last line is element wise. This is a convex optimization problem and can be solved by satisfying the Karush-Kuhn-Tucker (KKT) conditions (see, e.g, \cite{Chong-book}). Once the optimal $\mathbf{s}$ and $\mathbf{t}$ are found, we form $\mathbf{r_M}$ and then find $P$ and $Q$ from the signs and moduli of the elements of $\mathbf{r}_M$ as in \eqref{eq:pn}. 

\begin{example}\label{ex:maxsnr}
For $N=16$ and $M=8$, solving \eqref{eq:maxsnr} yields the following Max-SNR design: 
\begin{align}
P &= [ 0   \ 1 \    0 \  1 \   1 \     0 \    0 \  1 \   1  \   0 \    0 \   1 \   1 \    0 \   1 \    0],\nonumber\\  
%Q &= [0.0069 \  0.0429 \   0.0948 \   0.0623 \   0.0656 \   0.0770  \  0.0713 \   0.0792 \   0.0792  \  0.0713  \  0.0770  \  0.0656  \  0.0623  \  0.0948  \  0.0429  \  0.0069].\nonumber
Q &= 10^{-2}\times [0.69 \  4.29 \  9.48 \  6.23 \  6.56 \   7.70  \  7.13 \   7.92 \nonumber\\
& \quad\quad\quad\quad 7.92  \  7.13  \  7.70  \  6.56  \  6.23  \  9.48  \  4.29  \  0.69].\nonumber
\end{align}
\end{example}

\begin{remark}\label{rm:snrgain}
Let $S_Q(\theta)$ denote the spectrum of the Q sequence: 
\begin{equation}
S_{Q(\theta)}=\sum\limits_{n=0}^{N-1} q_n e^{jn\theta}.
\end{equation}
The {\em effective bandwidth} $\beta_Q$ (see, e.g., \cite{Stoica}) of this sequence is given by
\begin{align}
\beta_{Q} &=\frac{\frac{1}{2\pi}\int_{-\pi}^{\pi}\mathcal{S}_{Q}(\theta)d\theta}{\mathcal{S}_{Q}(0)}\nonumber\\
&=\frac{\sum_{n=0}^{N-1}q_n^2}{\bigl(\sum_{n=0}^{N-1}q_n\bigr)^2}\nonumber\\
&= \frac{\Vert\mathbf{q}\Vert_2^2}{\Vert\mathbf{q}\Vert_1^2}.
\end{align}
Therefore, we obtain  
\begin{equation}
\rho=\frac{L\sigma_b^2}{N_0\beta_{Q}},
\end{equation}
indicating that output SNR is inversely proportional to the effective bandwidth of $Q$. We can think of $1/\beta_Q$ as the SNR gain due to processing $N$ pulses together, because $L\sigma_b^2/N_0$ is the SNR from processing a single waveform.  
\end{remark}

\begin{remark}
In the case of a noisy radar return, detection of targets is inhibited
by spillage of energy coming from nearby bins as well as ``in-bin'' (measurement)
noise. Consider the case of two point targets, with equal powers $\sigma_b^2$, that are 
$\theta=\nu T$ apart in Doppler. The ratio
\begin{equation}\label{eq:rpq}\begin{split}
\kappa(\theta) &=\frac{\sigma_b^2|\chi_{P,Q}(0,0)|^2}{\sigma_b^2\max\limits_{k\neq 0}|\chi_{P,Q}(k,\theta)|^2+\eta} \\
& = \biggl(\gamma(\theta)^{-1}+\rho^{-1}\biggr)^{-1}
\end{split}
\end{equation}
characterizes the separability of these two targets in the noisy
environment. 
\end{remark}

\section{An Illustration}\label{sec:illust}

We consider a radar scene that contains three strong reflectors of equal
amplitude at different ranges and two weak targets (each $30\dB$ weaker)
with  small Doppler frequencies relative to the stronger
reflectors. The baseband waveforms $x(t)$ and $y(t)$ are generated by phase coding a
a raised cosine pulse $\Omega(t)$ with a pair length $L=64$ Golay complementary sequences, constructed as in Remark \ref{rm:golayconst}.\footnote{The specific choice of the Golay complementary pair of a given length does not have any noticeable effect on the results.} The chip interval is $T_c=100$ \textit{n}sec, and the PRI is $T=50$ $\mu$sec. 

Figure~1 illustrates the annihilation of range-sidelobes around the
zero-Doppler axis for three different length $N=16$ $(P,Q)$ designs and
compares their delay-Doppler responses with that of a conventional
design: an alternating transmission of Golay complementary waveforms followed by a matched filter at the
receiver. The horizontal and vertical axes depict Doppler and
delay, respectively. Color bar values are in $\dB$. All four transmit pulse trains have the same total energy.  

In the conventional design, shown in Figure~1(a), the weak targets
are almost completely masked by the range sidelobes of the stronger
reflectors, whereas the PTM design, shown in Figure~1(b), clears the
range sidelobes inside a narrow Doppler interval around the
zero-Doppler axis. The order of the spectral null for range sidelobe
suppression in this case is $M=\left(\log_2 N \right)-1=3$. This
brings the range sidelobes below $-80\dB$ inside the $[-0.1,-0.1]$rad
Doppler interval and enables detection of the weak targets. 

If the difference in the Doppler frequencies of the weak and strong reflectors is larger,
a higher order null is needed  to annihilate the range sidelobes
inside a wider Doppler band. Figure~1(c) shows that the Binomial
design (of length $N=16$) can expand the cleared (below $-80\dB$) region
to $[-1,-1]$rad by creating a null of order $M=N-2=14$ at zero
Doppler. 

Figure~1(d) shows the delay-Doppler response of a $(P,Q)$
design that has the largest SNR among all $(P,Q)$ pairs (Max-SNR design) that achieve
an $(M=8)$th order spectral null at zero Doppler (see Example \ref{ex:maxsnr}). The cleared (below $-80$dB) region in this case is
the $[-0.5,0.5]$rad Doppler interval. %The Max-SNR pair is obtained by solving \eqref{eq:maxsnr} with $M=8$ and $N=16$.

\begin{figure}\label{f:PTM}
\begin{center}
\begin{tabular}{cc}
\subfigure[]
{\includegraphics[width=3.2in,height=1.93in]{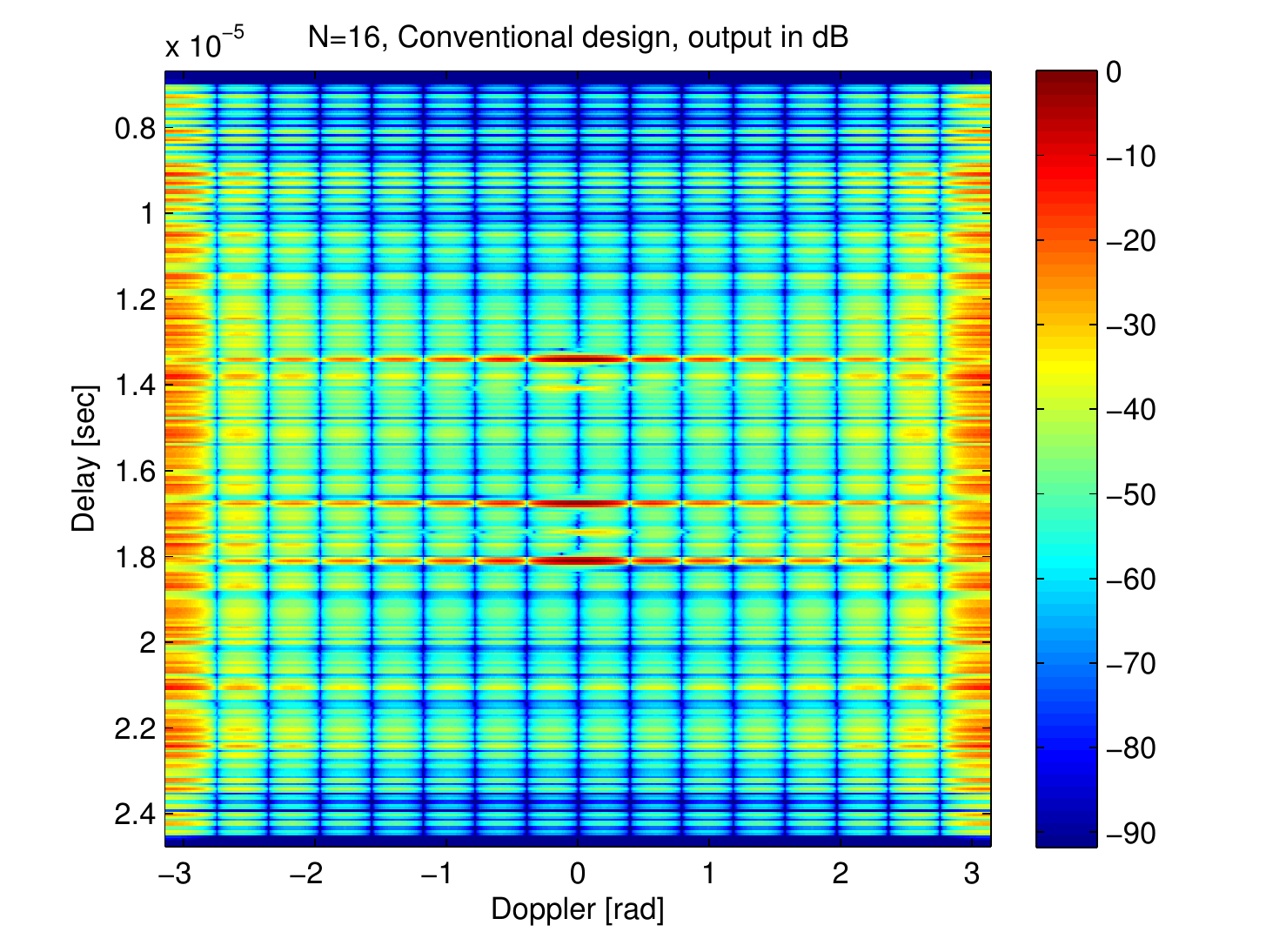}}\\
\subfigure[]
{\includegraphics[width=3.2in,height=1.93in]{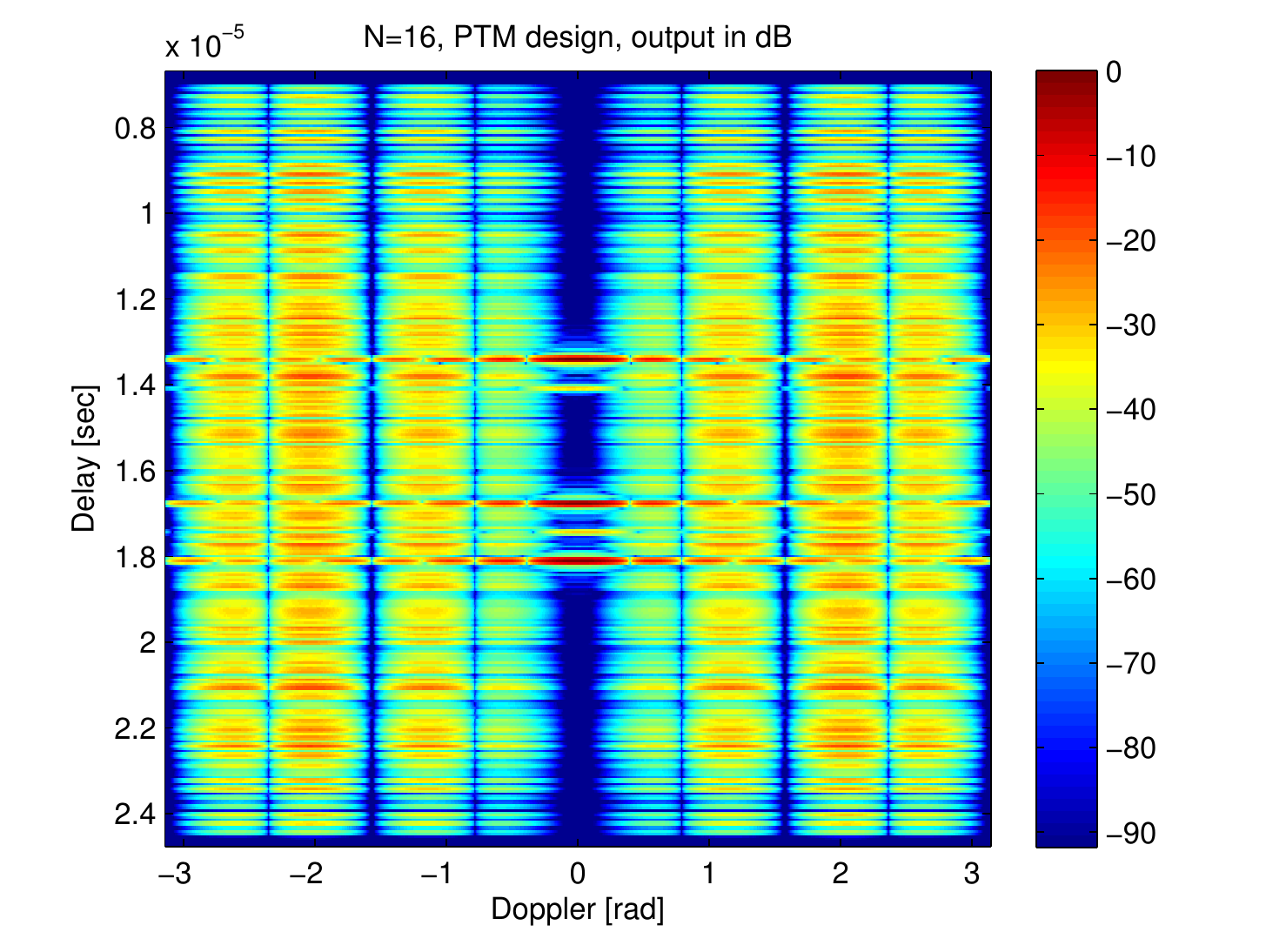}}\\
\subfigure[]
{\includegraphics[width=3.2in,height=1.93in]{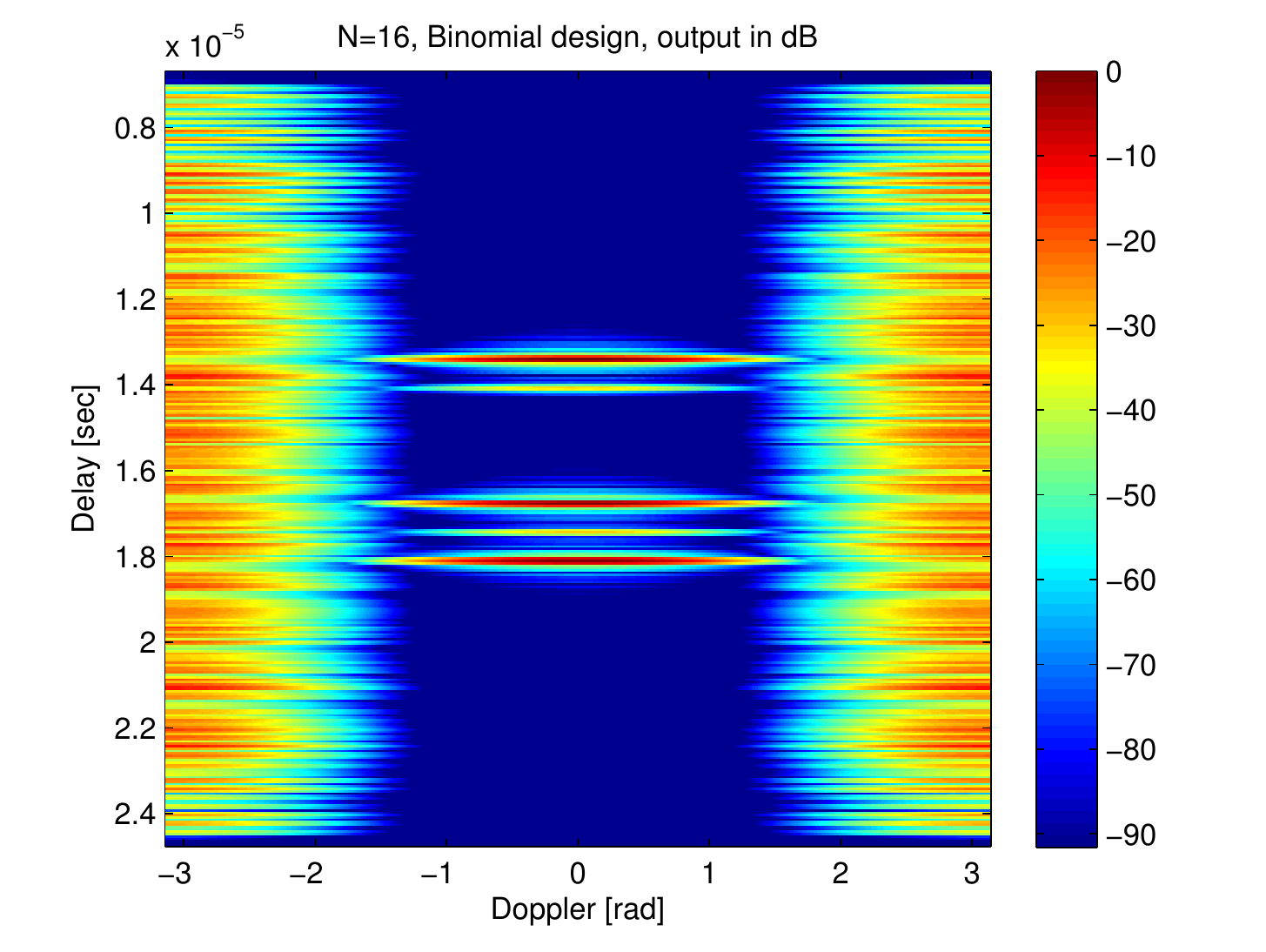}}\\
\subfigure[]
{\includegraphics[width=3.2in,height=1.93in]{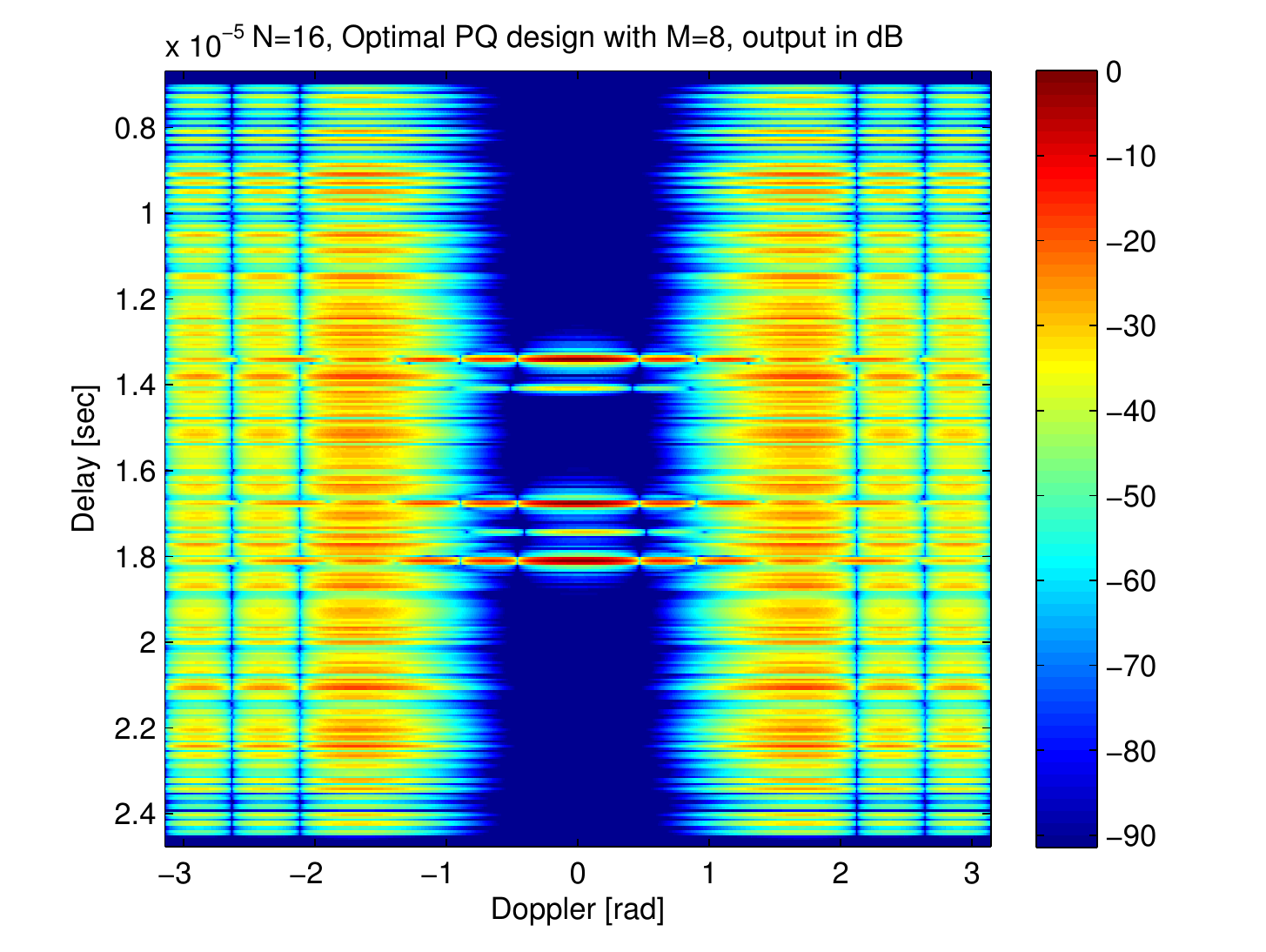}}\\
%(a) Conventional design & (b) PTM design & (c) Binomial design
\end{tabular}
\end{center}
\caption{Comparison of output delay-Doppler maps for different $(P,Q)$
  designs: (a) conventional design, (c) PTM design, (e) Binomial
  design, and (g) max-SNR design with an $8$-th order null. The scene contains three
  strong (equal amplitude) stationary reflectors at different ranges
  and two weak slow moving targets ($30\dB$ weaker).}
\end{figure}

Table~\ref{tab:1} compares the three designs in terms of the null
order and the SNR gain $1/\beta_{Q}$ (see Remark \ref{rm:snrgain}), and shows that, by joint design of the $P$
and $Q$ sequences, a null of relatively high order can be achieved
without considerable reduction in SNR compared to a conventional matched filter
design.
\begin{table}[h]
\caption{Null order \& SNR for different designs}\label{tab:1}
\begin{center}
\begin{tabular}{|c|c|c|}
  \hline
  $(P,Q)$ design & Null order & SNR gain\\ \hline
  Conventional & 0 & 16\\ \hline
  PTM  & 3 & 16\\ \hline
  Max-SNR with $M=8$ & 8 & 13.76\\ \hline
  Binomial & 14 & 6.92\\ \hline
\end{tabular}
\end{center}
\end{table}

\section{${(P,Q)}$ Pulse Trains for Larger Sets of Complementary Waveforms}\label{sc:multi}

So far, we have studied the design of ${(P,Q)}$ pulse trains for a
library consisting of only two complementary waveforms. We now extend
this construction to larger collections. Here we have  a set
of \textit{$D$-complementary} length-$L$ sequences
$\mathcal{X}=\{x_0,x_1,...,x_{D-1}\}$, where the  autocorrelations
$C_{z_d}[k]$ of the $z_d$ sequences satisfy
\begin{equation}\label{eq:golay_D}
\sum_{d=0}^{D-1}C_{z_d}[k]=DL\delta[k].
\end{equation}
Normally $D$ is a power of $2$,  and no pairwise complementarity is
assumed.  For example, for $D=4$, we can choose $x_0$, $x_1$, $x_2$,
and $x_3$ to form a Golay complementary quad, satisfying
(\ref{eq:golay_D}), without making $x_i$, $x_j$, $i\neq j$ Golay
complementary pairs. The reader is referred to \cite{Tseng-IT72} for the 
construction of Golay quads and larger sets of complementary sequences.

We will assemble pulse trains for transmission and filtering according
to sequences $P=\{p_n\}_{n=0}^{N-1}$ and $Q=\{q_n\}_{n=0}^{N-1}$,
in analogous fashion to the $2$-complementary case covered in earlier sections. To allow for
indexing of $D$ different waveforms,
we take $P$ to be a $D$-ary
sequence; that is,  defined over the alphabet $\mathcal{D}=\{0,1,...,D-1\}$;
At the $n$th PRI of the $P$ pulse train the waveform $x_d(t)$,
phase coded by $x_d[\ell]$ as in \eqref{eq:sum_of_chips}, is transmitted if $p_n=d$.
The ordering of the waveforms in the $Q$ pulse train is the same as that in the $P$
pulse train, but the $n$th waveform is weighted by $q_n$ as before. 

Let $\omega=e^{j2\pi /D}$. Note that, for each $d$ from $0$ to $D-1$, we have 
\begin{equation}
\frac{1}{D}\sum_{r=0}^{D-1}\omega^{r(p_n-d)}=\begin{cases}1, &\ p_{n}=d\\
0, & \ p_{n}\neq d \end{cases}.
\end{equation}
Then, the $P$ and $Q$ pulse trains, denoted again (with some abuse of notation) by $z_P(t)$ and $z_Q(t)$ can be expressed as 
 \begin{equation}\label{eq:zpv}
 z_P(t)=\sum\limits_{n=0}^{n-1}\left(\sum\limits_{d=0}^{D-1}\left(\frac{1}{D}\sum\limits_{r=0}^{D-1}\omega^{r(p_n-d)}\right)x_d(t-nT)\right)
 \end{equation}
and
 \begin{equation}\label{eq:zqv}
z_Q(t)=\sum\limits_{n=0}^{n-1}q_n\left(\sum\limits_{d=0}^{D-1}\left(\frac{1}{D}\sum\limits_{r=0}^{D-1}\omega^{r(p_n-d)}\right)x_d(t-nT)\right).
 \end{equation}
 
Following similar steps as those taken in deriving \eqref{eq:ptmamb}, we can write the discretized (in delay) cross-ambiguity function between $z_P(t)$ and $z_Q(t)$ as 
%\begin{equation}\label{eq:chipqd}
%\tilde{\chi}_{P,Q}(k,\theta)=\sum_{d=0}^{D-1}\biggl(\sum_{{\substack{n= 0 \\ p_{n}=
%d}}}^{N-1}q_{n}e^{jn\theta}\biggr)C_{z_d}[k].
%\end{equation}
%Let $\omega=e^{j2\pi /D}$. Note that, for each $d$ from $0$ to $D-1$, we have 
%\begin{equation}
%\sum_{r=0}^{D-1}\omega^{r(p_n-d)}=\begin{cases}D, &\ p_{n}=d\\
%0, & \ p_{n}\neq d \end{cases}.
%\end{equation}
%Thus, we can write \eqref{eq:chipqd} as  
\begin{equation}
    \label{eq:abdfkjbakjdb}
    \begin{split}
    &\chi_{P,Q}(k,\theta)\\
    &=\frac{1}{D}\sum_{d=0}^{D-1}C_{z_d}[k]\left(\sum_{n=0}^{N-1}q_{n}e^{jn\theta}
    \sum_{r=0}^{D-1}\omega^{r(p_{n}-d)}\right)\\
    &=\frac{1}{D}\sum_{r=0}^{D-1}\biggl(\sum_{n=0}^{N-1}
    \omega^{rp_{n}}q_{n}e^{jn\theta}\biggr)\biggl(\sum_{d=0}^{D-1}\omega^{-rd}C_{z_d}[k]\biggr)
    \\
&=\frac{1}{D}\biggl(DL\delta(k)\sum_{n=0}^{N-1}q_ne^{jn\theta}+\sum_{r=1}^{D-1}S_{P,Q,r}(\theta)\Delta_r\biggr),
\end{split}
\end{equation}
where
\begin{equation}
S_{P,Q,r}(\theta)=\sum_{n=0}^{N-1}\omega^{rp_{n}}q_{n}e^{jn\theta}
\end{equation}
and
\begin{equation}
\Delta_r=\sum_{d=1}^{D-1}\omega^{-rd}C_{z_d}[k].
\end{equation}

The first term on the right hand side of \eqref{eq:abdfkjbakjdb} is an impulse in delay and does not have range sidelobes. The second term has range sidelobes because of $\Delta_r$. To suppress the range sidelobes in (\ref{eq:abdfkjbakjdb}), it suffices to suppress the spectra $S_{P,Q,r}(\theta)$,  for $r=1,\ldots, D-1$. We note that for $D=2$, the term $S_{P,Q,1}(\theta)$ is the spectrum $S_{P,Q}(\theta)$ analyzed earlier.  

Consider the complex-valued functions on $[0,2\pi]\times \{1,\cdots, \omega^{D-1}\}$ of the form
\begin{equation}\label{eq:zetaform}
f(\theta, \zeta) = \sum_{n=0}^{N-1} q_n \zeta^{p_n} e^{j n\theta},
\end{equation}
where $q_n \geq 0$
and $p_n \in \{0,\cdots, D-1\}$.
The set of functions of this form is  denoted by $W_N(\omega)$. We note that the spectra $S_{PQ,r}(\theta)$, $r=1,2,\ldots,D-1$, are all elements of $W_N(\omega)$.

Higher order polynomials of the form \eqref{eq:zetaform} (longer
sequences) can be constructed as follows. For $i=1, 2$,
take $f_i \in W_{N_i}(\omega)$ with
\begin{equation}
f_i(\theta, \zeta_i) = \sum_{n=0}^{N_i-1} q_n^{(i)} \zeta_i^{p_n^{(i)}} e^{j n\theta},
\end{equation}
then 
\begin{multline*}
f_2(\theta, \zeta_2)f_1(N_1\theta, \zeta_1) = \\ \sum_{n=0}^{N_1-1}\sum_{\ell=0}^{N_2-1} q_\ell^{(2)}q_n^{(1)} \zeta_2^{p_\ell^{(2)}}\zeta_1^{p_n^{(1)}} e^{j (N_1\ell + n)\theta}\in W_{N_1N_2}(\omega),
\end{multline*}  
because $\zeta_2^{p_\ell^{(2)}}\zeta_1^{p_n^{(1)}} = \omega^{ap_\ell^{(1)}+bp_n^{(2)}}$, for some fixed $a,b \in \{0, \cdots, D-1\}$.
 
The following theorem presents a construction of a $(P,Q)$ pair of sequences
that render a higher order null for each of
$S_{P,Q,1}(\theta),...,S_{P,Q,D-1}(\theta)$:
\begin{theorem}\label{thm:Horder}
Fix a set of $D=2^m$ complementary sequences and set $\omega = e^{2\pi j/D}$. Suppose
that $f_1, \cdots, f_m$ are functions in $W_N(\omega)$, each with the property that  
$f_k(\cdot, -1)$ has an $M$th order spectral null at $\theta=0$. Then each
of the functions on $W_{N^m}$:
\begin{equation}\label{eq:SPQr}
S_{P,Q,r}(\theta) = \prod_{k=1}^m f_k(N^{k-1}\theta,
\omega^{2^{k-1}r}),  
\end{equation}
for $r=1,2,\ldots, D-1$, has an $M$-th order 
spectral null at $\theta=0$. 
\end{theorem}

\begin{proof}
  First note that, for $k=1,\cdots, m$, $f_k(N^{k-1}\theta, -1)$
  has an $M$th  order spectral null at $\theta=0$,
  because $f_k(\theta,-1)$ does. Now, for $1\leq r \leq 2^m-1$,
  by  prime factorization, $r = 2^{m-k'}\ell$,
  for some odd positive integer $\ell$,
  and some $1\leq k' \leq m$,
  in which case $\omega^{2^{k'-1}r}=-1$.
  Thus, for each $1\leq r \leq 2^m-1$,
  $S_{P,Q,r}(\theta)$
  has a factor with an $M$th
  order spectral null at $\theta=0$ and the theorem follows.
\end{proof}

This theorem provides a method for constructing pulse trains with
$M$th order nulls at $\theta=0$ from functions $f\in W_N(\omega)$
with the property that $f_k(\cdot, -1)$ has an $M$th
order spectral null at $\theta=0$. We need a method for constructing
the latter. We know from Theorem~\ref{thm:2a} that $\{g_n(\theta)=(1-e^{j\theta})^n |\; n=M+1, \ldots, N-1\}$ provides a basis for the subspace $T_M$ of functions having an $M$-th order spectral null at $\theta=0$. Consider the function 
\begin{equation}
f_k(\theta,\zeta)=(1+\zeta e^{j\theta})^n h(\theta,\zeta), \ n=M+1,M+2,\ldots,N-1, 
\end{equation}
where $h(\theta,\zeta) \in W_{N_k-n}(\omega)$. This function is in $W_{N_k}(\omega)$ by construction, and it has the property that $f_k(\cdot,-1)$ has an $M$th order null at $\theta=0$ because $g_n(\theta)=(1-e^{j\theta})^n$ is a factor of $f_k(\cdot,-1)$.

Thus, elements of $W_N(\omega)$ are of the form
\begin{equation}
 f(\theta, \zeta) = \sum_{n=M+1}^{N-1} a_n g_n(\theta, \zeta)
\end{equation}
for $a_n \geq 0$, where 
\begin{equation}
g_n(\theta, \zeta) = (1+\zeta e^{i\theta})^n
\end{equation}
all have the required property that $f(\theta, -1)$ has an $M$-th
order spectral null at $\theta=0$. 
Alternatively, we note that taking any element of $T_M$, 
\begin{equation}
f(\theta) = \sum_{n=0}^{N-1} q_n (-1)^{p_n} e^{jn\theta}, 
\end{equation}
with $p_n\in \{0,1\}$, and replacing $(-1)$ by $\zeta$ to obtain
\begin{equation}\label{eq:firstorder}
f(\theta, \zeta) = \left(\sum_{\substack{n=0\\p_n=0}}^{N-1} q_n e^{jn\theta}\right) + \left(\sum_{\substack{n=0\\p_n=1}}^{N-1} q_n e^{jn\theta}\right) \zeta 
\end{equation}
also gives an element of $W_N(\omega)$, where $f(\theta, -1)$ has  an $M$-th order spectral null at $\theta=0$.

\begin{example}\label{ex:akjbd}
	Given a set of $D=4$ complementary sequences  $\{z_0,z_1,z_2,z_3\}$
	(called a Golay complementary quad), take $f_k(\theta,\zeta) = g_3(\theta,\zeta)$ for $k=1, 2$.  We
	know  that $g_3(\theta,-1)$ creates a second-order null of range sidelobes
	at $\theta=0$ for a complementry pair. Now 
	generate the length-16 sequences $P$ and $Q$  according to Theorem~\ref{thm:Horder}, that is, 
	\begin{align}\label{eq:ex3}
		S_{P,Q,r} &= g_3(\theta,j^r)\,g_3(4\theta, (-1)^r)\nonumber\\ 
		&= (1+j^r e^{j\theta})^3 (1+(-1)^r e^{j\theta})^3 
	\end{align}
	for $r=0,\cdots, 3$, which gives
	\begin{align*}
	P&=\{p_{n}\}_{n=0}^{15} = 0\ 1\ 0\ 1\ 2\ 3\ 2\ 3\ 0\ 1\ 0\ 1\ 2\ 3\ 2\
	3,\\
	Q&=\{q_{n}\}_{n=0}^{15} = 1\ 3\ 3\ 1\ 3\ 9\ 9\ 3\ 3\ 9\ 9\ 3\ 1\ 3\ 3\ 1.
	\end{align*}	
	The transmit and receive waveforms ($\mathbf s$ and $\mathbf w$,
	respectively)  can be vectorized as
	\begin{align*}
	\mathbf{z}_{P}&=[z_0\ z_1\ z_0\ z_1\ z_2\ z_3\ z_2\ z_3\ z_0\ z_1\ z_0\ z_1\ z_2\ z_3\ z_2\ z_3],\\
	\mathbf{z}_{Q}&=[z_0\ 3z_1\ 3z_0\ z_1\ 3z_2\ 9z_3\ 9z_2\ 3z_3\ 3z_0\
                9z_1\ 9z_0\ 3z_1\ \\&  \qquad \qquad z_2\  3z_3\ 3z_2\ z_3],
	\end{align*}
	where in the corresponding continuous-time waveforms $z_{P}(t)$ and $z_Q(t)$ consecutive elements in the above vectors are separated by a PRI $T$.  	
	Writing out \eqref{eq:ex3} in detail, we have 
	\begin{equation*}\begin{split}
	S_{P,Q,1}(\theta)&=(1+je^{j\theta})^3\underline{(1-e^{j4\theta})^3},\\
	S_{P,Q,2}(\theta)&=\underline{(1-\phantom{j}e^{j\theta})^3}(1+e^{j4\theta})^3,\\
	S_{P,Q,3}(\theta)&=(1-je^{j\theta})^3\underline{(1-e^{j4\theta})^3},
	\end{split}\end{equation*}
  Each of $S_{P,Q,r}(\theta)$,
  ($r=1,2,3)$
  has a second-order null at $\theta=0$,
  resulting from the underlined factor.  Finally we note that using
  $f_k(\theta, \zeta) = 1+3je^{j\theta}+3e^{2j\theta}+je^{3j\theta}$,
  according to \eqref{eq:firstorder}, gives precisely the same transmit
  and receive sequences.
\end{example}

\section{Extension to MIMO Radar}\label{sc:MIMO}

We now extend the construction of $(P,Q)$ pairs to MIMO radar. We consider a MIMO radar with an array of $2^K$, $K\ge 1$, transceiver elements and construct pulse trains of complementary waveform vectors and receive filter banks for which the cross-ambiguity matrix is essentially free of range sidelobes inside an interval around the zero-Doppler axis.

\textit{Definition 5:} Complementary vector sets \cite{Tseng-IT72}. A set of $D$ sequence-valued vectors $\bx_d$, $d=0,\ldots, D-1$, each composed of $D$ length-$L$ unimodular sequences, is called complementary if 
\begin{equation}
\sum\limits_{d=0}^{D-1}\bC_{\bx_d}[k]=D L\bI_{D}\delta[k]
\end{equation} 
where $\bI_{D}$ is the $D\times D$ identity matrix and 
\begin{equation}
\bC_{\bx_d}[k]=\sum\limits_{\ell=-(L-1)}^{L-1} \bx_d[\ell]\bx_d[\ell-k]^H, \ d=0,\ldots, D'-1, 
\end{equation}
is the autocorrelation matrix of $\bx_d[\ell]$ at lag $k$. 

It has been shown in \cite{Tseng-IT72} that such complementary sets can be constructed when $D=2^K$ for $K> 1$ in a recursive fashion. The reader is referred to \cite{Tseng-IT72} for details of such constructions. We discuss a special case of the construction of complementary sets in an example shortly.

\begin{remark}
Let $\bS_D$ be a $D\times D$ sequence-valued matrix whose columns $\bx_d$, $d=0,\ldots, D-1$ form a complementary set. Then, it is straightforward to show that $\bS_D$ is {\em paraunitary}, that is, 
\begin{equation}
\sum\limits_{\ell=0}^{L-1}\bS_D[\ell]\bS_D[\ell-k]^H=\sum\limits_{d=0}^{D-1} \bC_{\bx_d}[k]=DL\bI_{D}\delta[k].
\end{equation}
Conversely, the columns of a paraunitary matrix form a complementary set. This paraunitary property is the same as the one in the theory of paraunitary filter banks and quadrature mirror filters (see, e.g., \cite{Vet-book}), where sequences are thought of as Finite-Impulse-Response (FIR) filters. In fact, complementary sets are special cases of paraunitary filter banks, where the FIR tap coefficients are unimodular.       
\end{remark}

\begin{example}
Suppose $x$ and $y$ are length-$L$ Golay complementary sequences. Consider the matrix 
\begin{equation}\label{eq:S2}
\bS_{2}=\begin{bmatrix} x \ \hfill & \ \hfill -\tilde{y}\\ y \ \hfill & \ \hfill \tilde{x}\end{bmatrix}
\end{equation}
where $\tilde{x}$ and $\tilde{y}$  are reversed versions of $x$ and $y$, respectively, that is, $\tilde{x}[\ell]=x[L-1-\ell]^*$ and $\tilde{y}[\ell]=y[L-1-\ell]^*$, $\ell=0,\ldots,L-1$. Then, the columns of $\bS_{2}$ are complementary ($D=2$). This is because $\bS_{2}$ is paraunitary:
\begin{align}
\sum\limits_{\ell=0}^{L-1}\bS_{2}[\ell]\bS_2^H[\ell]&=\begin{bmatrix} C_x[k]+C_y[k] \ & \ C_{xy}[k]-C_{xy}[k] \\ C_{yx}[k]-C_{yx}[k] \ & \ C_x[k]+C_y[k]\end{bmatrix}\nonumber
\\&=2L\bI_2\delta[k].
\end{align}
Larger paraunitary matrices, or equivalently larger complementary sets can be constructed recursively. Let $\bS_{2^{k-1}}$ be a $2^{k-1}\times 2^{k-1}$ paraunitary matrix. Then  
\begin{equation}\label{eq:S2k}
\bS_{2^k}=\begin{bmatrix} \bS_{2^{k-1}} \ \hfill& \ \hfill \bS_{2^{k-1}} \\ \widetilde{\bS}_{2^{k-1}} \ \hfill & \ \hfill -\widetilde{\bS}_{2^{k-1}}\end{bmatrix}, 
\end{equation}
where $\widetilde{\bS}_{2^{k-1}}$ signifies the reversal of all of the sequences in the matrix $\bS_{2^{k-1}}$, is a $2^k \times 2^k$ paraunitary matrix. This can be easily verified by an induction argument.  
\end{example}

Now consider a paraunitary matrix $\bS_D$. Let $\bX$ denote a waveform-valued matrix whose $(m,n)$th element, $x_{m,n}(t)$, is obtained by phase coding the basic pulse shape $\Omega(t)$ with the $(m,n)$ element of $\bS_D$ for $m=0,\ldots,D-1$ and $n=0,\ldots, D-1$. Each column of $\bX$ is a waveform vector whose elements are transmitted simultaneously by the radar array. Let $\bx_{d}$ denote the $d$th column of $\bX$, that is, $\bx_d=[x_{0,d}(t),\ldots,x_{D-1,d}(t)]^T$, $d=0,\ldots, D-1$. Let $P=\{p_n\}_{n=0}^{N-1}$ be a $D$-ary sequence of length $N$ over the alphabet $\mathcal{D}=\{0,1,\ldots,D-1 \}$ and $Q=\{q_n\}_{n=0}^{N-1}$ be a nonnegative sequence of length $L$. 

We will assemble pulse train vectors $\bz_P(t)$ and $\bz_Q(t)$ for, respectively, transmission and filtering by selecting and weighting the waveform vectors $\bx_{d}$, $d=0,1,\ldots,D-1$ according to sequences $P$ and $Q$ in a similar fashion as the $D$-ary case covered in Section \ref{sc:multi}. At the $n$th PRI of $\bz_{P}(t)$, the waveform vector $\bx_d(t)$ is transmitted if $p_n=d$. The ordering of the waveform vector in $\bz_{Q}(t)$ is the same as that in $\bz_{P}(t)$, but the $n$th waveform vector is weighted by $q_n$. The expressions for  $\bz_P(t)$ and $\bz_Q(t)$ are similar to \eqref{eq:zpv} and \eqref{eq:zqv}, respectively, with waveforms vectors $\{\bx_d\}_{d=0}^{D-1}$ replacing the scalar waveforms  $\{x_d\}_{d=0}^{D-1}$.

Transmitting $\bz_{P} (t)$ and filtering the return by (correlation with) $\bz_{Q}(t)$ results in a matrix-valued point-spread function (in delay and Doppler) that is given by the cross-ambiguity matrix $\bChi_{PQ}(\tau,\nu)$ between $\bz_P (t)$ and $\bz_P (t)$:
\begin{equation}\label{eq:crossmat}
\bchi_{PQ}(\tau,\nu)=\int\limits_{-\infty}^{\infty} \bz_{P}(t)\bz_{Q}(t)^H e^{-j\nu t} dt.
\end{equation}
This cross-ambiguity matrix is $D\times D$. The $d$th diagonal element of $\bchi_{PQ}(\tau,\nu)$ is the ambiguity function of $x_d(t)$ an its $(n,m)$th ($n\neq m$) off-diagonal element is the cross-ambiguity function between $x_n(t)$ and $x_m(t)$.   

%Following similar steps and assumptions as those made in deriving \eqref{eq:} and \eqref{eq:}, we can discretize the cross-ambiguity matrix $\bchi_{PQ}(\tau,\nu)$ in delay and express it as 
%\begin{equation}\label{eq:mchipqd}
%\chi_{P,Q}(k,\theta)=\sum_{d=0}^{D-1}\biggl(\sum_{{\substack{n= 0 \\ p_{n}=
%d}}}^{N-1}q_{n}e^{jn\theta}\biggr)\bC_{x_d}[k],
%\end{equation}
%where $\theta=\nu T$. This is the multi-channel (MIMO) counterpart of the cross-ambiguity function in \eqref{eq:}, where autocorrelations $C_{x_d}$ of complementary sequences $x_d$, $d=0,1,\ldots,D-1$ are replaced by the autocorrelation matrices $\bC_{\bx_d}$ of complementary vector sequences $\bx_d$, $d=0,1,\ldots,D-1$. 

Following similar steps as those taken in deriving \eqref{eq:abdfkjbakjdb}, we can discretize the cross-ambiguity matrix $\bchi_{PQ}(\tau,\nu)$ in delay and express it as 
\begin{eqnarray}\label{eq:mchipqd2}
\bchi_{PQ}(k,\theta)&=\frac{1}{D}\biggl(DL\sum_{n=0}^{N-1}q_ne^{jn\theta}\biggr)\bI_D\delta[k]\nonumber\\
&+\frac{1}{D}\biggl(\sum_{r=1}^{D-1}S_{P,Q,r}(\theta)\bDelta_r\biggr),
\end{eqnarray}
where $\theta=\nu T$, 
\begin{equation}
S_{P,Q,r}(\theta)=\sum_{n=0}^{N-1}\omega^{rp_{n}}q_{n}e^{jn\theta}, 
\end{equation}
with $\omega=e^{j2\pi/D}$, and
\begin{equation}
\bDelta_r=\sum_{d=1}^{D-1}\omega^{-rd}\bC_{x_d}[k].
\end{equation}
This is the multi-channel (MIMO) counterpart of the cross-ambiguity function in \eqref{eq:abdfkjbakjdb}, where autocorrelations $C_{x_d}$ of complementary sequences $x_d$, $d=0,1,\ldots,D-1$ are replaced by the autocorrelation matrices $\bC_{\bx_d}$ of complementary vector sequences $\bx_d$, $d=0,1,\ldots,D-1$. 

The first term on the right hand side of \eqref{eq:mchipqd2} is an impulse in delay times a factor of identity, and is therefore free of range sidelobes. 
in delay and does not have range sidelobes. The second term has range sidelobes because of $\bDelta_r$. But the size of the entries of $\bDelta_r$ are all controlled by the spectrum $S_{P,Q,r}(\theta)$, which we studied in Section \ref{sc:multi}. The entries of $\bDelta_r$, and subsequently their weighted sum,  can be suppressed by creating high-order spectral nulls in $S_{P,Q,r}(\theta)$, $r=1,2,\ldots, D-1$ as stated in Theorem \ref{thm:Horder} . Thus, by selecting $P$ and $Q$ sequences according to  Theorem \ref{thm:Horder}, we can construct a cross-ambiguity matrix that is a factor times the identity at zero delay and vanishes at all nonzero delays in an interval around the zero Doppler axis.  

\begin{remark}
In the special case of $D=2$, the design of $P$ and $Q$ sequences follow the discussions of Section \ref{sc:PQbinary}. 
\end{remark}

\begin{example}\label{ex:2by2}
Let us consider the simplest MIMO case, where we have a MIMO radar with $D=2$ collocated transceivers. The paraunitary matrix used in phase coding in this case is $\bS_2$, given in \eqref{eq:S2}. In this case, we have two complementary waveform vectors $\bx_1(t)$ and $\bx_2(t)$, and the control of range sidelobes is similar to that in Section \ref{sc:PQbinary} and the three examples considered there are applicable here as well. Here the cross-ambiguity matrix $\bchi_{PQ}(k,\theta)$ is two-by-two. Figures~\ref{fig:2by2}(a)-(d) show the magnitude of the first diagonal element of the cross-ambiguity matrix for (a) the conventional design, (b) the PTM design, (c) the binomial design, and (d) the max-SNR design.  As can be seen, the ambiguities corresponding to the PTM, the binomial, and the max SNR design are all essentially delta functions in delay (range) in a Doppler interval around the zero Doppler axis. The plots for the second diagonal elements are identical and are not shown separately. Figures~\ref{fig:2by2}(e)-(h) show the magnitude of the first off-diagonal element of the cross-ambiguity matrix for (e) the conventional design, (f) the PTM design, (g) the binomial design, and (h) the max-SNR design. As can be seen, these cross-ambiguity functions corresponding to the PTM, the binomial, and the max-SNR designs all vanish in an interval around the zero Doppler axis. The magnitude of the second off-diagonal element is identical to that of the first off-diagonal elements, because of the cross-ambiguity matrix in \eqref{eq:crossmat} has Hermitian symmetry, because of the ordering of Golay waveforms in $\bz_P$ and $\bz_Q$ is the same.         
\end{example}

\begin{figure*}[tp]
\begin{center}
\begin{tabular}{cccc}
\subfigure[]{\includegraphics[width=.22\textwidth]{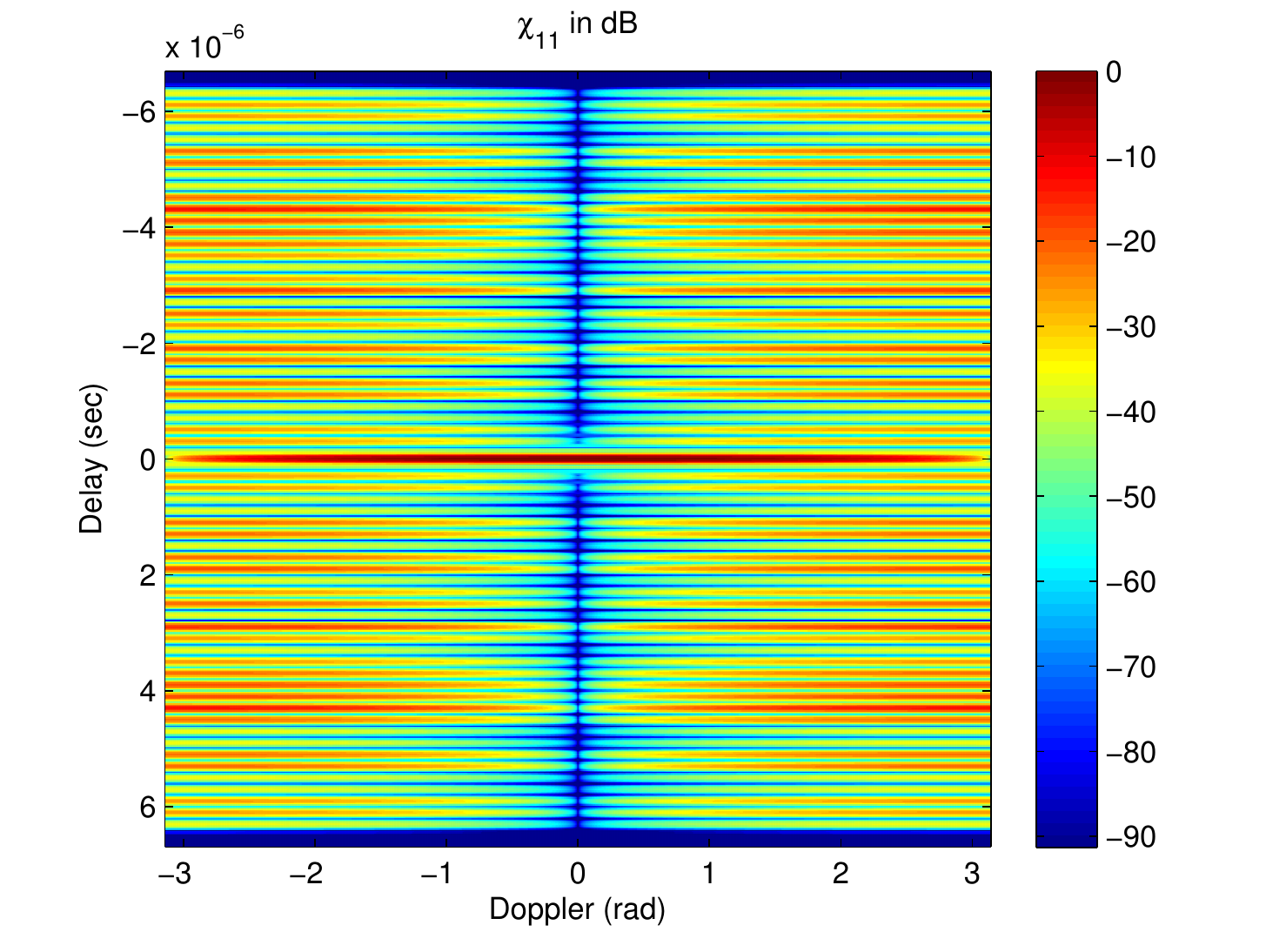}} & 
\subfigure[]{\includegraphics[width=.22\textwidth]{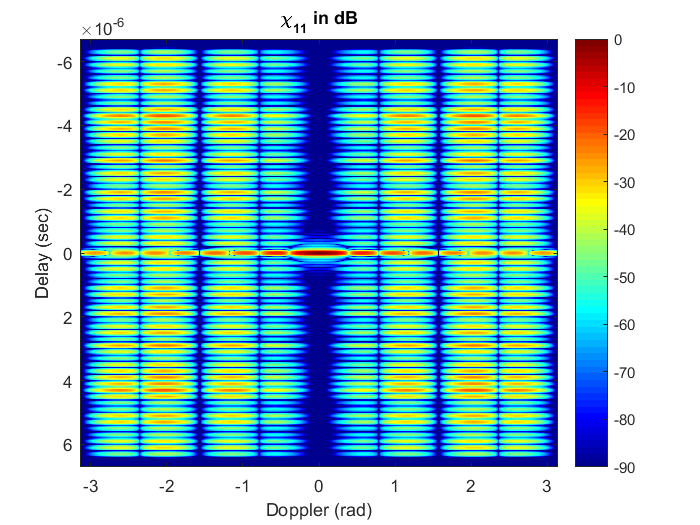}}&
\subfigure[]{\includegraphics[width=.22\textwidth]{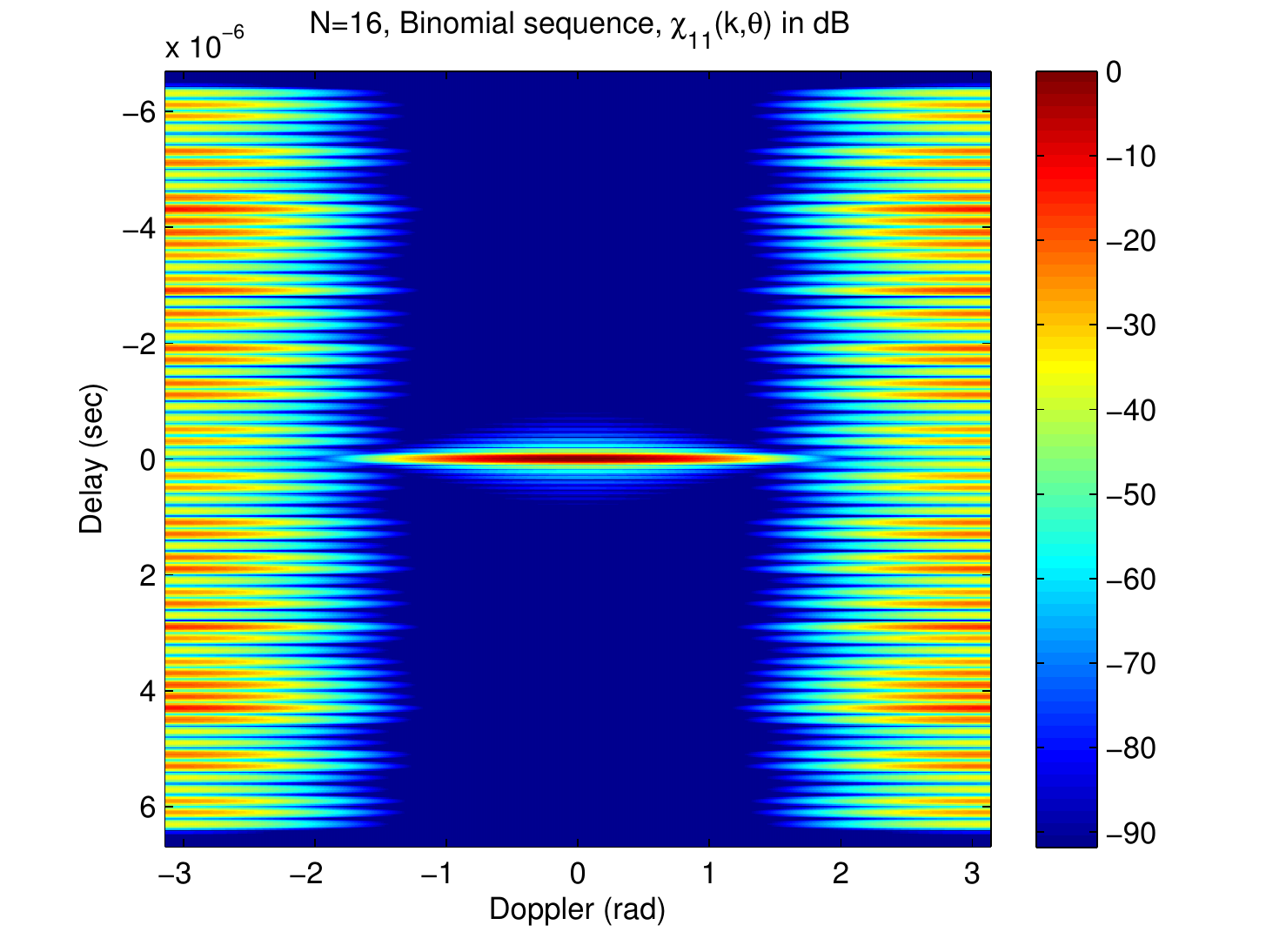}}&
\subfigure[]{\includegraphics[width=.22\textwidth]{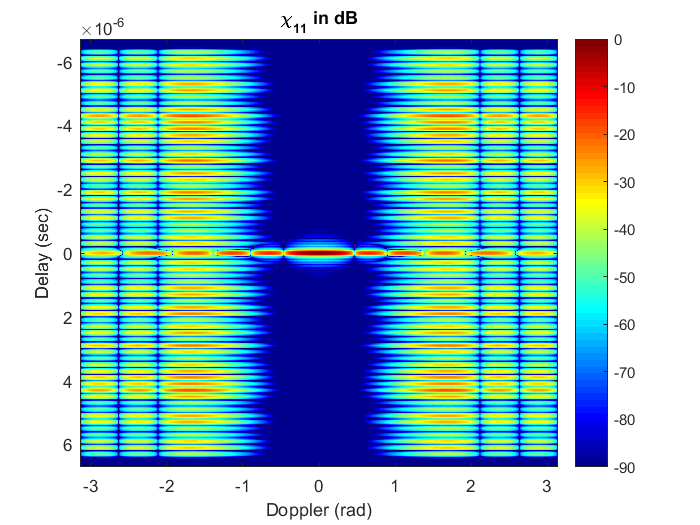}}\\
(a) Conventional (diagonal) & (b) PTM (diagonal) & (c) Binomial (diagonal) & (d) Max-SNR (diagonal)\\
\subfigure[]{\includegraphics[width=.22\textwidth]{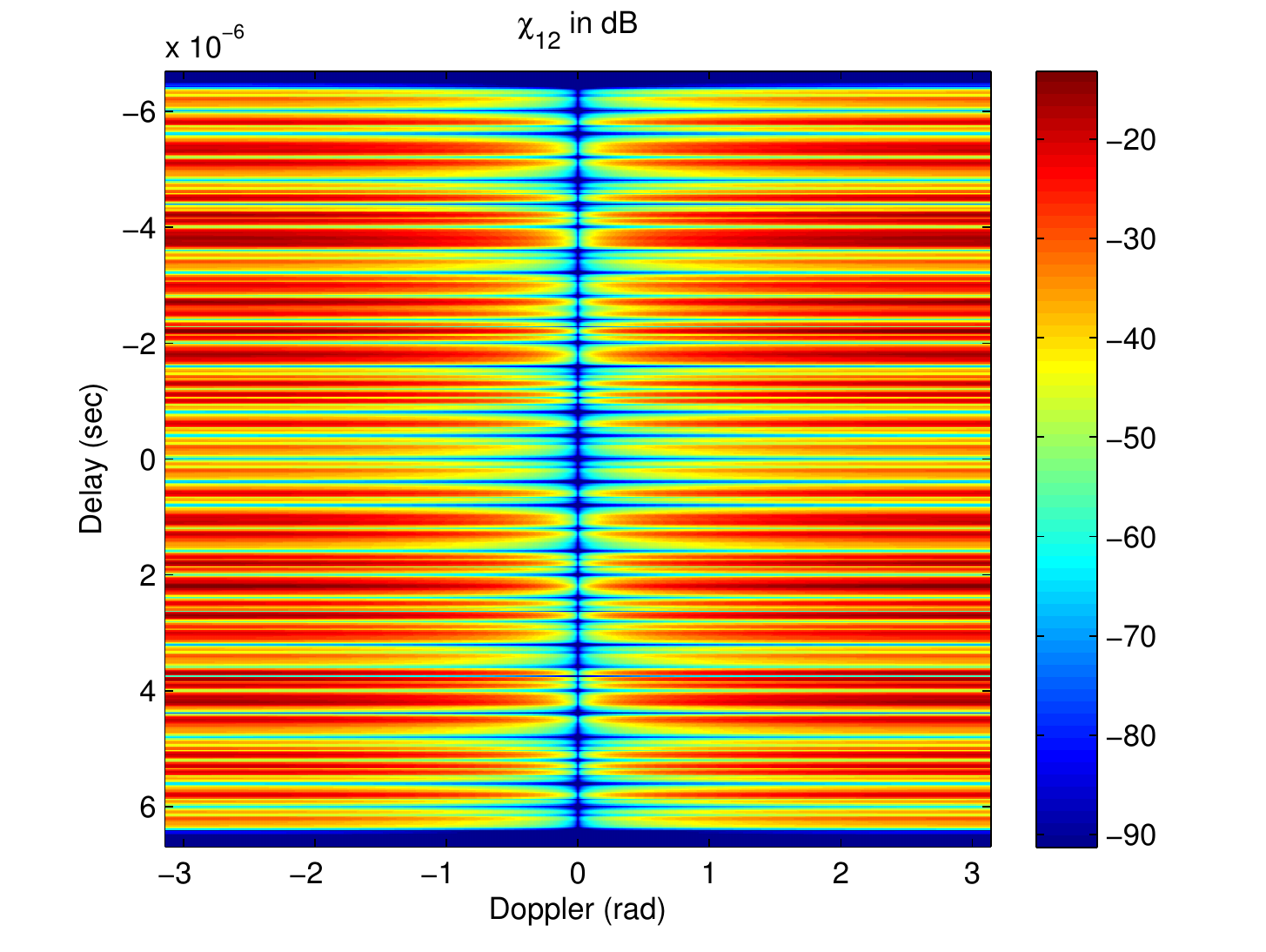}} & 
\subfigure[]{\includegraphics[width=.22\textwidth]{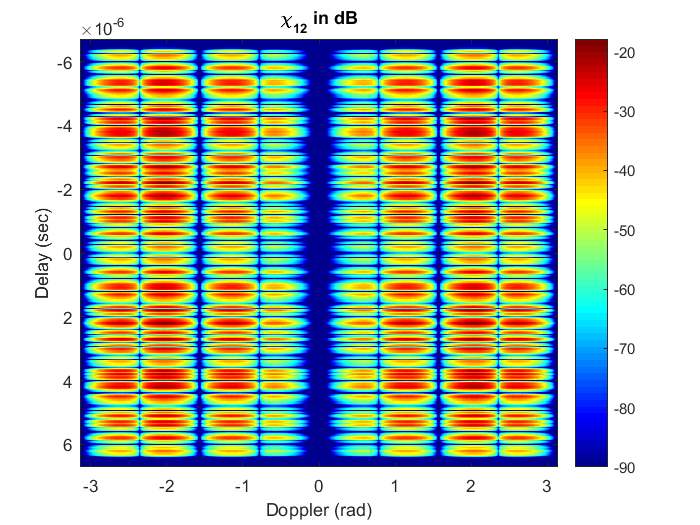}}&
\subfigure[]{\includegraphics[width=.22\textwidth]{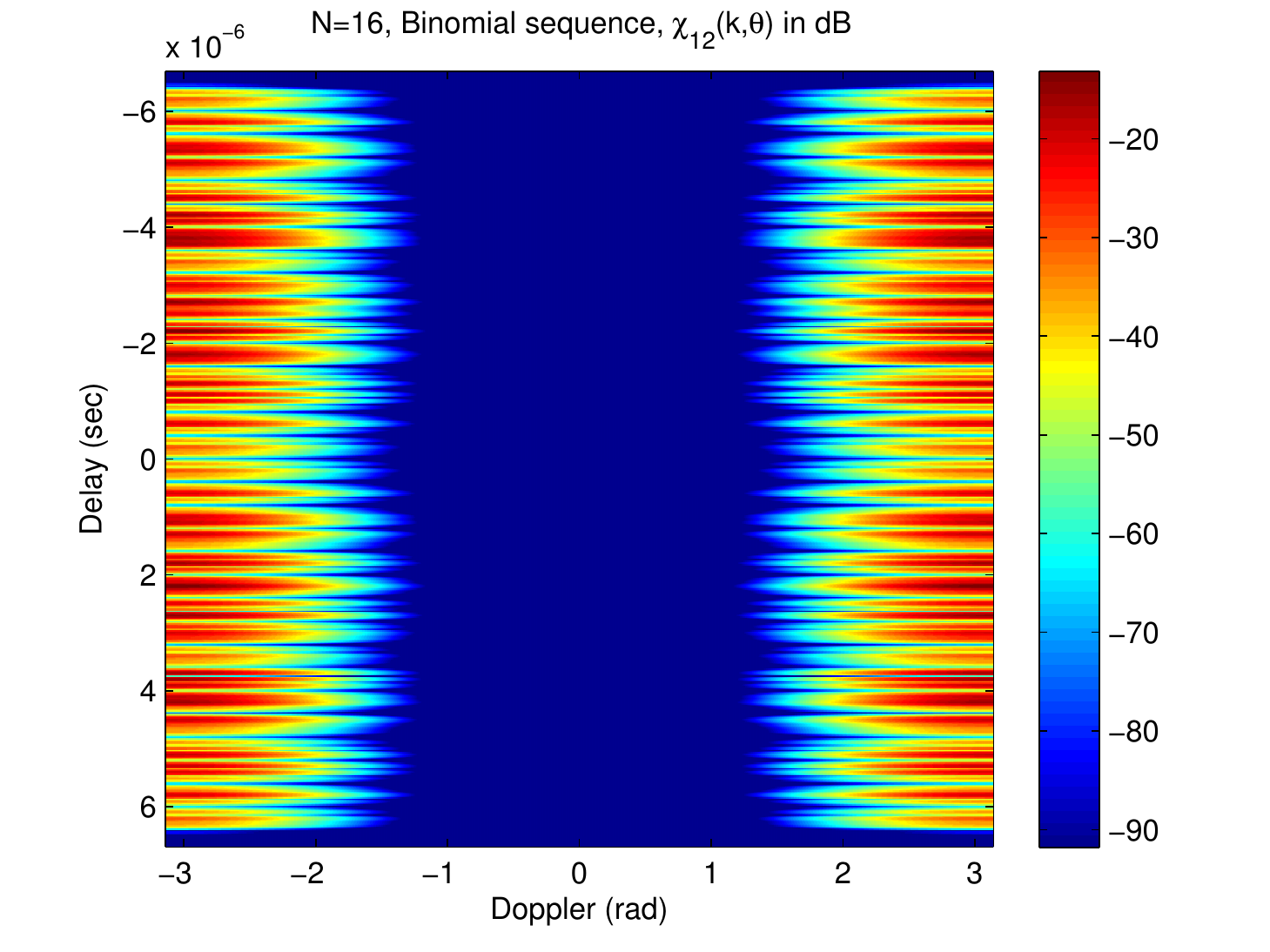}}&
\subfigure[]{\includegraphics[width=.22\textwidth]{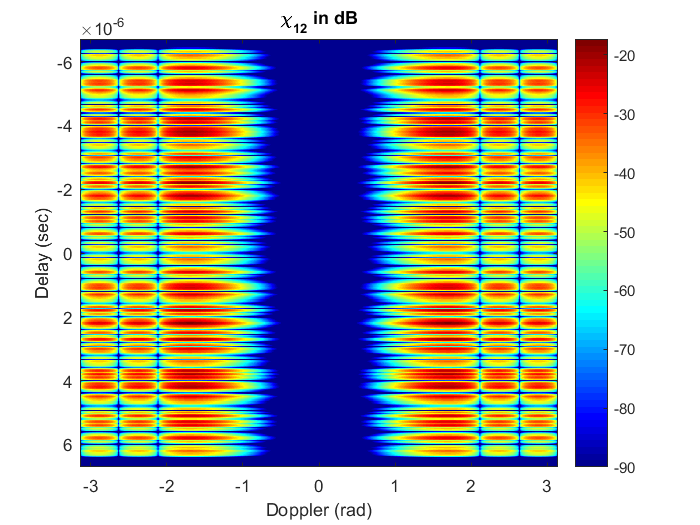}}\\
(e) Conventional (off-diagonal) & (f) PTM (off-diagonal) & (g) Binomial (off-diagonal) & (h) Max-SNR (off-diagonal)\\
\end{tabular}
\end{center}
\caption{Magnitudes of the diagonal and off-diagonal elements of the $2$-by-$2$ cross-ambiguity matrix, in Example \ref{ex:2by2}, for the conventional (plots (a) and (e)), the PTM (plots (b) and (f)), (c) the binomial (plots (c) and (g)), and the max-SNR (plots (d) and (h)) designs.}\label{fig:2by2}
\end{figure*}
\begin{figure*}[h]
\begin{center}
\begin{tabular}{cccc}
\subfigure[]{\includegraphics[width=.22\textwidth]{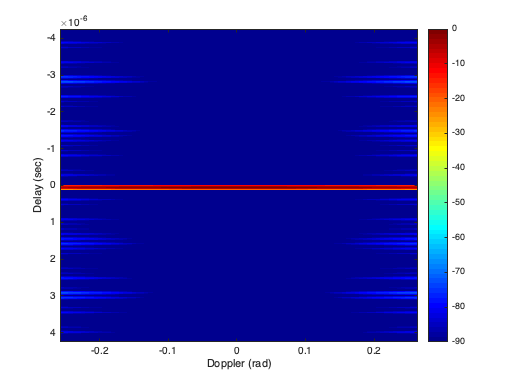}} & 
\subfigure[]{\includegraphics[width=.22\textwidth]{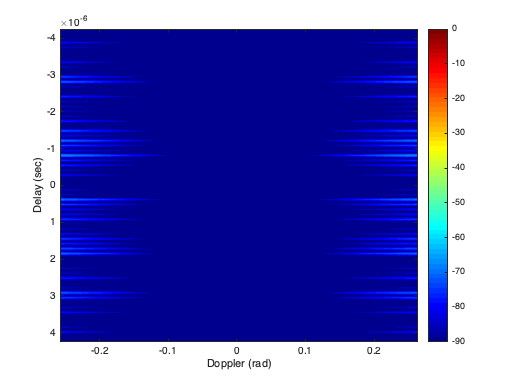}}&
\subfigure[]{\includegraphics[width=.22\textwidth]{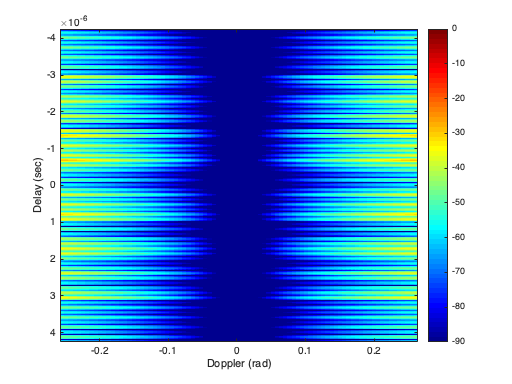}}&
\subfigure[]{\includegraphics[width=.22\textwidth]{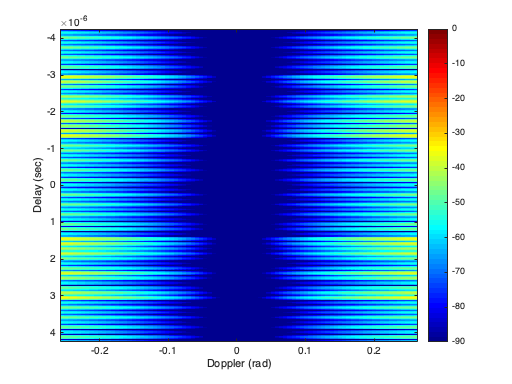}}\\
(a) $| \left[ \bchi_{PQ}\right]_{11} |$ & (b) $| \left[ \bchi_{PQ}\right]_{12} |$ & (c) $| \left[ \bchi_{PQ}\right]_{13} |$ & (d) $| \left[ \bchi_{PQ}\right]_{14} |$\\
\subfigure[]{\includegraphics[width=.22\textwidth]{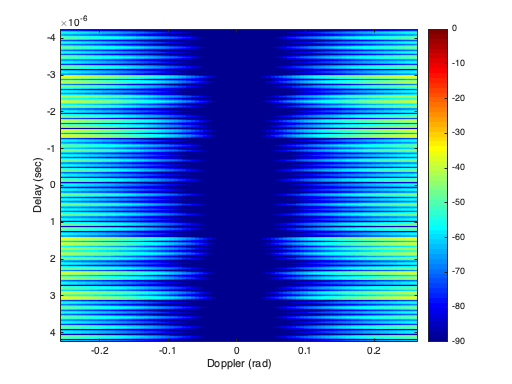}} & 
\subfigure[]{\includegraphics[width=.22\textwidth]{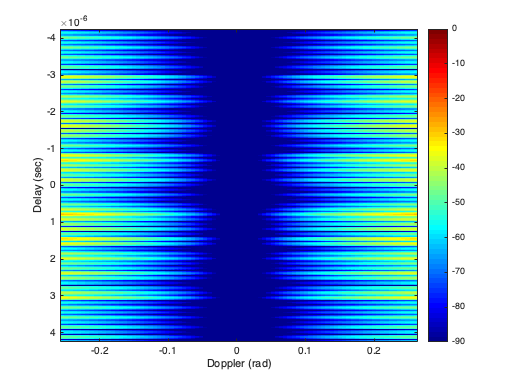}}&
\subfigure[]{\includegraphics[width=.22\textwidth]{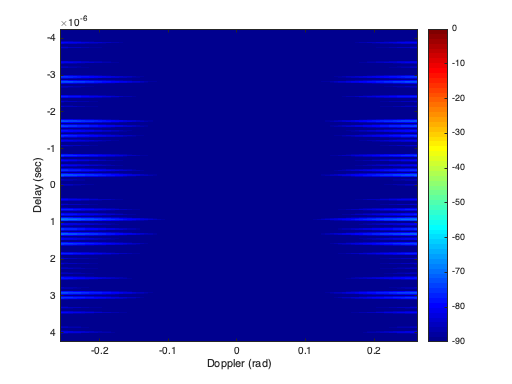}}&\\
(e) $| \left[ \bchi_{PQ}\right]_{23} |$ & (f) $| \left[ \bchi_{PQ}\right]_{24} |$ & (g) $| \left[ \bchi_{PQ}\right]_{34} |$ & 
\end{tabular}
\end{center}
\caption{Magnitudes of the elements of the $4$-by-$4$ cross-ambiguity matrix for Example \ref{ex:4by4}. Only the first diagonal element and the off-diagonal elements in the upper triangle of the cross-ambiguity matrix are shown in the figure, because the cross-ambiguity matrix in \eqref{eq:crossmat} is Hermitian symmetric and has identical diagonal elements by construction. The plots are shown only in the Doppler interval $(-\pi/12,\pi/12]$ to highlight sidelobe suppression in range around zero Doppler.}\label{fig:4by4}
\end{figure*}

\begin{example}\label{ex:4by4}
We now consider a MIMO case with $D=4$ collocated transceivers. The paraunitary matrix used in phase coding in this case is $\bS_4$, constructed as in \eqref{eq:S2k} with $k=2$. In this case, we have four complementary waveform vectors $\{\bx_d(t)\}_{d=0}^{3}$, and the control of range sidelobes is similar to that in Section \ref{sc:multi} and the example presented there is applicable here as well. Here the cross-ambiguity matrix $\bchi_{PQ}(k,\theta)$ is four-by-four. We wish for the diagonal elements of the cross-ambiguity matrix to look like a delta function in delay (range) in an interval around the zero Doppler axis and for the off-diagonal elements to vanish in an interval around the zero-Doppler axis. 

Figures~\ref{fig:4by4}(a)-(d) illustrate these effects for the $P$ and $Q$ designed in Example \ref{ex:akjbd} in Section \ref{sc:multi}. The figure shows the magnitude of the first diagonal element of the cross-ambiguity matrix and the magnitudes of the off-diagonal elements in the upper triangle of the cross-ambiguity matrix.  The magnitudes of the other diagonal elements are identical to that of the first diagonal element and hence are not shown separately. The plots for the other off-diagonal elements (in the lower triangle) are identical to the ones shown, because of the Hermitian symmetry of the cross-ambiguity matrix due to the identical ordering of Golay waveforms in $\bz_P$ and $\bz_Q$. The plots are shown only in the Doppler interval $(-\pi/12,\pi/12]$ to highlight sidelobe suppression in range around zero Doppler. Outside this interval the cross-ambiguity elements have large sidelobes, similar to what is observed in previous examples.        
\end{example}

\section{Conclusion}

In this paper we have presented a general approach to the
construction of radar
transmit-receive pulse trains with cross-ambiguity functions that are free
of range sidelobes inside an interval around the zero Doppler axis. The transmit
pulse train is constructed by a binary sequence $P$ that codes
the transmission of a pair of Golay complementary waveforms across
time. For the receiver pulse train each waveform is weighted by some
integer according to an integer sequence $Q$. The range sidelobes of
the cross-ambiguity function are shaped by the spectrum of essentially the product of $P$ and $Q$. By properly choosing the sequences $P$ and $Q$, 
the range sidelobes can be significantly reduced  inside an interval around the zero Doppler axis. A general way for constructing such sequences has been presented, by specifying the subspace (along with a basis) for sequences that have spectral nulls of a given order around zero Doppler.  The output signal-to-noise ratio (SNR) of $(P,Q)$  pairs depends only on the choice of $Q$. By jointly designing the transmit-receive sequences $(P,Q)$, we can maximize SNR subject to achieving a given order of the spectral null. A detailed comparison of two special cases of ${(P,Q)}$ pulse train design: PTM and Binomial design has been presented. 

We also have demonstrated that, for a larger set of complementary sequences, the desired $P$ and $Q$ sequences can be derived from an extension of the joint design of $P$ and $Q$ sequences for a Golay complementary pair.

We have also extend the construction of $(P,Q)$ pairs to multiple-input-multiple-output (MIMO) radar, by designing transmit-receive pairs of paraunitary waveform matrices whose matrix-valued cross-ambiguity function is essentially free of range sidelobes inside a Doppler interval around the zero-Doppler axis.

\bibliographystyle{IEEEtran}
\bibliography{reference,MIMORadar,DoppRef}%,linalgebbib}
\end{document}